\documentclass[twoside,11pt]{article}
\usepackage{jmlr2e}
\usepackage{graphicx,  amssymb, array, amsmath, float}
\usepackage{color}
\usepackage{placeins} % gives FloatBarrier to fix position of Figures
\usepackage{bm}
\usepackage{float}
\usepackage{mathtools}
\usepackage{comment}

\usepackage{enumitem} %tweak itemized spacing
\usepackage{hyperref}
\newtheorem{assumption}{Assumption}

\newcommand{\ind}{\stackrel{\mathrm{ind}}{\sim}}
\newcommand{\iid}{\stackrel{\mathrm{iid}}{\sim}}

\DeclarePairedDelimiterX{\infdivx}[2]{(}{)}{%
  #1\;\delimsize\|\;#2%
}
\newcommand{\infdiv}{\infdivx}
\newcommand\smallO{
  \mathchoice
    {{\scriptstyle\mathcal{O}}}% \displaystyle
    {{\scriptstyle\mathcal{O}}}% \textstyle
    {{\scriptscriptstyle\mathcal{O}}}% \scriptstyle
    {\scalebox{.7}{$\scriptscriptstyle\mathcal{O}$}}%\scriptscriptstyle
  }

\DeclarePairedDelimiter\abs{\lvert}{\rvert}%
\makeatletter
\let\oldabs\abs
\def\abs{\@ifstar{\oldabs}{\oldabs*}}

\ShortHeadings{Bayesian Pseudo Posterior Mechanism under Asymptotic Differential Privacy}{Savitsky, Williams, Hu}
\firstpageno{1}

\begin{document}

\title{Bayesian Pseudo Posterior Mechanism under Asymptotic Differential Privacy}

\author{\name Terrance D.\ Savitsky \email Savitsky.Terrance@bls.gov \\
       \addr Office of Survey Methods Research\\
       U.S. Bureau of Labor Statistics\\
       2 Massachusetts Ave NE\\
       Washington, DC 20212, USA
       \AND
       \name Matthew R.\ Williams \email mrwillia@nsf.gov \\
       \addr National Center for Science and Engineering Statistics\\
       National Science Foundation\\
       2415 Eisenhower Ave\\
       Alexandria, VA 22314, USA
       \AND
       \name Jingchen Hu \email jihu@vassar.edu \\
       \addr Vassar College\\
       124 Raymond Ave, Box 27\\
       Poughkeepsie, NY 12604, USA}

\editor{}

\maketitle

\begin{abstract}%
We propose a Bayesian pseudo posterior mechanism to generate record-level synthetic databases equipped with an $(\epsilon,\pi)-$ probabilistic differential privacy (pDP) guarantee, where $\pi$ denotes the probability that any observed database exceeds $\epsilon$.  The pseudo posterior mechanism employs a data record-indexed, risk-based weight vector with weight values $\in [0, 1]$ that surgically downweight the likelihood contributions for high-risk records for model estimation and the generation of record-level synthetic data for public release. The pseudo posterior synthesizer constructs a weight for each data record using the Lipschitz bound for that record under a log-pseudo likelihood utility function that generalizes the exponential mechanism (EM) used to construct a formally private data generating mechanism.  By selecting weights to remove likelihood contributions with non-finite log-likelihood values, we guarantee a finite local privacy guarantee for our pseudo posterior mechanism at every sample size.  Our results may be applied to \emph{any} synthesizing model envisioned by the data disseminator in a computationally tractable way that only involves estimation of a pseudo posterior distribution for parameters, $\theta$, unlike recent approaches that use naturally-bounded utility functions implemented through the EM.  We specify conditions that guarantee the asymptotic contraction of $\pi$ to $0$ over the space of databases, such that the form of the guarantee provided by our method is asymptotic. We illustrate our pseudo posterior mechanism on the sensitive family income variable from the Consumer Expenditure Surveys database published by the U.S. Bureau of Labor Statistics. We show that utility is better preserved in the synthetic data for our pseudo posterior mechanism as compared to the EM, both estimated using the same non-private synthesizer, due to our use of targeted downweighting.  
\end{abstract}


\section{Introduction}
\label{sec:intro}

Privacy protection is an important research topic, which attracts attention from government statistical agencies and private companies alike.  A commonly-used data privacy approach generates synthetic data from statistical models estimated on closely-held, private data for proposed release by statistical agencies \citep{Rubin1993synthetic} and \citep{ Little1993synthetic}. This approach replaces the closely-held (by the statistical agency) database with multiple synthetically generated record-level databases. The synthetic databases are released to the public who would use them to conduct any analyses of which they would conceive to be conducted on the real, confidential record-level data.  The synthetic data approach replaces multiple queries performed on a summary statistic with the publication of the synthetic databases encoded with privacy protection, making this approach independent of the specific queries performed by users or putative intruders.

\subsection{Differential Privacy}
Our focus metric for measuring the relative privacy guarantee of our pseudo posterior synthesizing data mechanism introduced in the sequel is differential privacy \citep{Dwork:2006:CNS:2180286.2180305}.  We next provide a  definition for differential privacy \citep{McSherryTalwar2007}.

\begin{definition}[Differential Privacy]\label{def:DP}
Let $\mathbf{x}$ be a database in input space $\mathcal{X}^{n}$, where $\mathcal{X}^{n}$ denotes a space of databases of size (number of observations) $n$. Let $\mathcal{M}$ be a randomized mechanism  such that $\mathcal{M}(): \mathcal{X}^{n} \rightarrow O$. Then $\mathcal{M}$ is $\epsilon$-differentially private if
\[
\frac{Pr[\mathcal{M}(\mathbf{x}) \in O]}{Pr[\mathcal{M}(\mathbf{y}) \in O]} \le \exp(\epsilon),
\]
for all possible outputs $O = Range(\mathcal{M})$ under all possible pairs of datasets $\mathbf{x} \in \mathcal{X}^{n}$ where $\mathbf{y} \in \mathcal{X}^{n-1}$ differs from $\mathbf{x}$ by deleting one record or datum (under a leave-one-out (LOO) distance definition).
\end{definition}

Differential privacy is a property of the mechanism or data generating process and a mechanism that meets the definition above is guaranteed to be $\epsilon-$ differentially private, or $\epsilon-$ DP.  Differential privacy is called a ``formal" privacy guarantee because the $\epsilon-$ level or guarantee is independent of the behavior of a putative intruder seeking to re-identify the data and the guarantee is not lessened by the existence of other data sources that may contain information about the same respondents included in $\mathcal{X}^{n}$.

Differential privacy assigns a disclosure risk for a statistic to be released to the public, $f(\mathbf{x})$ (e.g., total employment for a state-industry) of any $\mathbf{x} \in \mathcal{X}^{n-1}$ based on the global sensitivity, $\Delta = \mathop{\sup}_{\mathbf{x}\in\mathcal{X}^{n},\mathbf{y}\in\mathcal{X}^{n-1}: ~\delta(\mathbf{x},\mathbf{y})=1}\abs{f(\mathbf{x}) - f(\mathbf{y})}$, over the space of databases, $\mathcal{X}$, where $\delta(\mathbf{x},\mathbf{y})$ denotes the number of records omitted from $\mathbf{x}$ in database, $\mathbf{y}$.  The distance metric, $\delta(\mathbf{x},\mathbf{y})$ denotes the LOO distance such that $\mathbf{x}$ differs from $\mathbf{y}$ by a single record, which is equivalent to using a Hamming-1 distance in the case of count based statistics of binary data records. If the value of the statistic, $f$, expresses a high magnitude change after the deletion of a data record in $\mathbf{y}$, then the mechanism will be required to induce a relatively higher level of distortion to $f$.  The more sensitive is a statistic to the change of a record, the higher its disclosure risk.

Our focus in this paper is where the mechanism, $\mathcal{M}$, is a model parameterized by $\theta$ from which replicate data are synthesized under an $\epsilon-$ DP guarantee.  A common approach for generating parameter draws of $\theta$ under the statistical model for synthesizing data is the exponential mechanism (EM) of \citet{McSherryTalwar2007}, which inputs a non-private mechanism for $\theta$ and generates $\theta$ in such a way that induces an $\epsilon-$DP guarantee on the overall mechanism. The EM is conditioned on the availability of a global sensitivity over the space of databases, $\Delta_{u}$ for some utility function, $u(\mathbf{x}, \theta)$, defined on the space of databases and the space of parameters, globally. 

\begin{comment} %old version
\begin{definition}(Exponential Mechanism)
The exponential mechanism releases values of $\theta$ from a distribution proportional to
\begin{equation}
\exp \left(\frac{\epsilon \, u(\mathbf{x}, \theta)}{2 \Delta_{u}} \right),
\end{equation}
where $u(\mathbf{x}, \theta)$ is a utility function, $\Delta_{u} = \mathop{\sup}_{\mathbf{x}\in \mathcal{X}^{n}} \,\, \mathop{\sup}_{\mathbf{x}, \mathbf{y}: \delta(\mathbf{x}, \mathbf{y}) = 1} \, \, \mathop{\sup}_{\theta \in \Theta} \,\, \abs{u(\mathbf{x}, \theta) - u(\mathbf{y}, \theta)}$ is the sensitivity, defined globally over $\mathbf{x} = (x_{1},\ldots,x_{n}) \in \mathcal{X}^n$, the $\sigma-$algebra of datasets, $\mathbf{x}$, governed by product measure, $P_{\theta_{0}}$; $\delta(\mathbf{x}, \mathbf{y}) = \# \{i: x_i \neq y_i\}$ is the Hamming distance between $\mathbf{x}, \mathbf{y} \in \mathcal{X}^n$.  Each draw of $\theta$ from the exponential mechanism satisfies $\epsilon-$DP, where $\epsilon$ is a budget target supplied by the publishing statistical agency.
\end{definition}
\end{comment}

\begin{definition}
(Exponential Mechanism)
The exponential mechanism releases values of $\theta$ from a distribution proportional to,
\begin{equation}
\exp \left(u(\mathbf{x}, \theta) \right),
\end{equation}
where $u(\mathbf{x}, \theta)$ is a utility function. Let \newline 
$\Delta_{u} = \mathop{\sup}_{\mathbf{x}\in \mathcal{X}^{n}} \,\, \mathop{\sup}_{\mathbf{x}, \mathbf{y}: \delta(\mathbf{x}, \mathbf{y}) = 1} \, \, \mathop{\sup}_{\theta \in \Theta} \,\, \abs{u(\mathbf{x}, \theta) - u(\mathbf{y}, \theta)}$ be the sensitivity, defined globally over $\mathbf{x} = (x_{1},\ldots,x_{n}) \in \mathcal{X}^n$, the $\sigma-$algebra of datasets, $\mathbf{x}$, governed by product measure, $P_{\theta_{0}}$ and the LOO distance metric, $\delta(\mathbf{x},\mathbf{y}) = 1$.  Then each draw of $\theta$ from the exponential mechanism is guaranteed to be $ \epsilon = 2\Delta_{u}-$DP.
\end{definition}
\noindent This result is based on the following definition of differential privacy under utility function, $u(\mathbf{x},\theta)$.
\begin{definition}
(Differential Privacy under the Exponential Mechanism)
A utility function, $u$, indexed by random parameters, $\theta$, gives $\epsilon-$differential privacy if for all databases, $\mathbf{x} \in \mathcal{X}^{n}$ and associated databases, $\mathbf{y}:\delta(\mathbf{x}, \mathbf{y}) = 1$,
and all parameter values, $\theta \in\Theta$,
\begin{equation}
\mbox{Pr}\left(u(\mathbf{x},\theta) \in O\right)\leq \exp(\epsilon) \times \mbox{Pr}\left(u(\mathbf{y},\theta) \in O\right),
\end{equation}
where $O = \mbox{range}(u)$.
\end{definition}

In order to set an arbitrary $\epsilon \ne 2\Delta_{u}$, we must modify the utility function $u(\mathbf{x}, \theta)$. The statistical agency owning the closely-held data will typically desire to determine $\epsilon$ as a matter of policy and not leave it to be $ \epsilon = 2\Delta_{u}$. The simplest and most common approach is to rescale it: $u^{*}(\mathbf{x}, \theta)=\frac{\epsilon}{2 \Delta_{u}}u(\mathbf{x}, \theta)$ \citep[See][among many others]{McSherryTalwar2007, Dwork:2006:CNS:2180286.2180305}. 

The EM inputs a utility function and its sensitivity constructed as the supremum of the utility over the space of databases, $\mathcal{X}^{n}$, and simultaneously, the parameter space, $\Theta$. \citet{WassermanZhou2010} and \citet{SnokeSlavkovic2018PSD} construct utility functions (e.g., the Kolmogorov-Smirnov distance between the empirical distributions of the real and synthetic datasets) that are \emph{naturally} bounded over all $\mathbf{x} \in \mathcal{X}^{n}$, resolving the challenge of using the potentially unbounded log-likelihood as the utility function.  Although the use of a naturally bounded utility resolves the issue of truncating the data and parameter spaces, there is a \emph{large}, and perhaps \emph{intractable}, computational cost to the use of these naturally bounded utilities to draw samples of $\theta$ from the distribution constructed from the EM; for example, \citet{SnokeSlavkovic2018PSD} must compute their $pMSE$ utility statistic multiple times for each proposed value, $\hat{\theta}_l$ ($l = 1, \cdots, L$), under a Metropolis-Hastings algorithm used to draw samples under the EM.  Furthermore, they assume the existence of some synthesizing distribution, $g(\hat{\theta})$, from which to draw synthetic data, is needed to compute their $pMSE$. In practice, $g$ will be defined as the posterior predictive distribution, $g(\mathbf{X} \mid \mathbf{x}, \hat{\theta}_l)$, which means the posterior distribution must be repeatedly estimated for \emph{each} draw from of $\theta$ from the EM.

\citet{Dimitrakakis:2017:DPB:3122009.3122020} utilize the log-likelihood as the utility under the EM such that the EM reduces to the the model posterior distribution, $\xi(\theta \mid \mathbf{x})$, as the mechanism, $\mathcal{M}$.  They specify conditions under which the Bayesian posterior distribution is a formally private mechanism for generating  synthetic data. They construct their posterior distribution from the private data, $\mathbf{x}$, and parameters, $\theta$, used to generate the synthetic data. They show that if the log-likelihood is Lipschitz continuous with bound $\Delta$ over the space of databases, $\mathbf{x}\in\mathcal{X}^{n}$ (the space of databases of size, $n$) and the space of parameters, $\theta\in\Theta$, then the posterior mechanism achieves an $\epsilon = 2 \Delta -$DP guarantee for each posterior draw of $\theta$, the model parameter(s). However, \citet{Dimitrakakis:2017:DPB:3122009.3122020} acknowledge that computing a finite $\Delta$, in practice, under the use of the log-likelihood is particularly difficult for an unbounded parameter space. They specify relatively simple Bayesian probability models where the Lipschitz bound is analytically available.  Even in such simple model setting, \citet{Dimitrakakis:2017:DPB:3122009.3122020} require truncation of the support of the prior distribution to achieve a finite $\Delta$. Relatively simply-constructed differentially private Bayesian synthesizers are similarly proposed by \citet{OnTheMap2008, AbowdVilhuber2008PSD, McClureReiter2012TDP} and \citet{BowenLiu2020}. The utility performance to preserve the real data distribution in the simulated synthetic data of these simple posterior mechanisms under a truncated prior support may be severely compromised by truncation and over-smoothing (induced by simple, parametric prior distributions). 

\citet{HuSavitskyWilliams2021rebds} design a record-indexed weight $\alpha_i \in [0,1]$, which is set to be inversely proportional to their construction for the identification risk probability of record, $i$; a data record that expresses a relatively high probability of identification disclosure will receive a likelihood weight, $\alpha_{i}$, that is closer to $0$, while a data record with a low disclosure probability will receive a likelihood weight, $\alpha_{i}$, that is closer to $1$.  The vector weights $\bm{\alpha} = (\alpha_1, \cdots, \alpha_n)$ are subsequently applied to the likelihood function of all $n$ records to form the pseudo posterior,
\begin{equation}
\label{eq:pp}
\xi^{\bm{\alpha}}\left(\theta \mid \mathbf{x},\gamma\right) \propto \left[\mathop{\prod}_{i=1}^{n}p\left(x_{i} \mid \theta\right)^{\alpha_{i}}\right]\xi\left(\theta\mid \gamma\right),
\end{equation}
where $\theta$ denotes the model parameters, $\gamma$ denotes the model hyperparameters and $\xi(\cdot)$ denotes the prior distribution. This construction employs a data record-indexed, risk-based weight vector with weights $\in [0, 1]$ to surgically downweight high-risk records in estimation of a pseudo posterior distribution for $\theta$, subsequently used to generate and release a synthetic record-level database. \citet{HuSavitskyWilliams2021rebds} show that this selective downweighting of records reduces the average of by-record risks as compared to an unweighted synthesis, while inducing only a minor reduction in utility. Their risk measure is based on a calculated probability of identification for a record. They cast a radius around the true data value for each record and count the number of record values that lie outside of the radius, which directly measures the extent that the target record is isolated and, therefore, easier for an intruder to discover by random guessing.  While this risk measure appeals to intuition, it is based on an assumption about the behavior of a putative intruder.  By contrast, the DP framework makes no explicit assumptions about the behavior or knowledge of an intruder.

This paper extends \citet{HuSavitskyWilliams2021rebds} and \citet{Dimitrakakis:2017:DPB:3122009.3122020} by constructing weights $\bm{\alpha}(\mathbf{x}) = \left(\alpha_{1}(\mathbf{x}),\ldots,\alpha_{n}(\mathbf{x})\right)$ where $\alpha_{i}(\mathbf{x}) \propto 1/\mathop{\sup}_{\theta\in\Theta}\vert f(x_{i}\mid\theta)\vert$ that we show achieves a formal privacy guarantee that they were not able to achieve.

The remainder of the paper is organized as follows: Section \ref{sec:theory} introduces a pseudo posterior mechanism using weights $\bm{\alpha}$ plugged into Equation~\ref{eq:pp}. We  generalize \citet{Dimitrakakis:2017:DPB:3122009.3122020} by establishing a direct functional association between a Lipschitz bound, $\Delta_{\bm{\alpha}}$, for the pseudo posterior mechanism and a $(\epsilon=2\Delta_{\bm{\alpha}})-$ DP guarantee. In Section \ref{sec:lipschitz}, we describe the computation details to produce a matrix of (absolute values for) log-likelihoods estimated for the $n$ records and $S$ parameter draws taken from the unweighted posterior distribution and their subsequent use to formulate a vector of record-indexed weights, $\bm{\alpha}$, for a single observed database, $\mathbf{x}$.  We then discuss the procedure to use the $\bm{\alpha}$ to estimate the pseudo posterior distribution and the computation of the Lipschitz bound for the pseudo posterior mechanism based on the observed or local database.  We call a Lipschitz bound constructed from a single (observed) database as a local Lipschitz. By contrast, we label a Lipschitz guarantee that represents a uniform bound over the space of databases as a global Lipschitz. Section \ref{sec:lockin} specifies formal conditions that guarantee the asymptotic contraction of a local Lipschitz bound to the global Lipschitz bound over the space of databases.  We include a Monte Carlo simulation study that generates a collection of local databases and shows that the infimum and supremum of the local Lipschitz bounds collapse together to a global value as $n$ approaches $1000$.  Section \ref{sec:app} focuses on our application to synthesizing the family income variable of a sample from the Consumer Expenditure surveys administered by the U.S. Bureau of Labor Statistics (BLS).  This section presents the risk and utility curves of locally differentially private synthetic data generated under the proposed pseudo posterior mechanism, compared to the EM. We conclude with a discussion in Section \ref{sec:conclusion}.

\section{Differential Privacy for the Pseudo Posterior}
\label{sec:theory}

In this section, we specify the connection between achieving a global Lipschitz bound, $\Delta_{\bm{\alpha}}$, under our pseudo posterior mechanism of Equation~(\ref{eq:pp})  with weights $\bm{\alpha}$ and an $\epsilon-$ DP (or global DP) guarantee (over the space of databases). We further re-purpose a result from \citet{WassermanZhou2010} to extend a global DP guarantee for the mechanism generating parameters, $\theta$, to the pseudo posterior predictive mechanism for generating synthetic data that is based on integrating with respect to the globally DP privacy guaranteed pseudo posterior distribution mechanism (used to generate the model parameters).  After having shown that achievement of a global Lipschitz under our pseudo posterior mechanism produces a global DP privacy guarantee, we discuss constructing by-record weights used in our pseudo posterior mechanism that are designed to be inversely proportional to the (absolute value of) log-likelihood utilities computed over the parameter space.  The log-likelihood for each record represents its relative risk of identification disclosure for the record since it governs the Lipschitz bound that defines the sensitivity.  This construction of weights allows us to achieve a global Lipschitz (linked to a global DP guarantee) without data or parameter truncation.

\subsection{Preliminaries}
We begin by constructing the probability space, $(\Theta, \beta_{\Theta})$, equipped with prior distribution, $\xi(\theta)$.  Observe a database sequence, $\mathbf{x} = \left(x_{1},\ldots,x_{n}\right) \in \mathcal{X}^{n}$  under $x_{1},\ldots,x_{n} \ind P_{\theta_{0}}$, for some $\theta_{0}\in\Theta$, we formulate the pseudo likelihood,
\begin{equation}
\label{pseudolike}
p_{\theta}^{\bm{\alpha}}(\mathbf{x}) = \mathop{\prod}_{i=1}^{n} p_{\theta_i}(x_i)^{\alpha_i(\mathbf{x})},
\end{equation}
for each $\theta \in \Theta$ and $\mathbf{x}\in\mathcal{X}^{n}$.  The pseudo likelihood exponentiates likelihood contributions by $\bm{\alpha}(\mathbf{x}) = (\alpha_{1}(\mathbf{x}),\ldots,\alpha_{n}(\mathbf{x}))$, where $\alpha_{i}(\mathbf{x}) \in [0,1]$ denote weights that are constructed to be inversely proportional to the local identification disclosure risk for each observed dataset record. These weights are subsequently used to selectively downweight the likelihood contributions for records in proportion to the level identification disclosure risks that they express.  

Given the prior and pseudo likelihood, we construct the pseudo posterior distribution,
\begin{equation}
\label{pseudopost}
\xi^{\bm{\alpha}}(B \mid \mathbf{x}) = \frac{\int_{\theta \in B} p_{\theta}^{\bm{\alpha}}(\mathbf{x}) d \xi(\theta)}{\phi^{\bm{\alpha}}(\mathbf{x})} = \frac{\mathop{\int}_{\theta\in B}e^{-r_{n,\bm{\alpha}\left(\theta,\theta^{\ast}\right)}}d \xi(\theta)}{\int_{\theta\in \Theta}e^{-r_{n,\bm{\alpha}\left(\theta,\theta^{\ast}\right)}}d \xi(\theta)},
\end{equation}
where $\phi^{\bm{\alpha}}(\mathbf{x}) \overset{\Delta}{=} \int_{\theta \in \Theta} p_{\theta}^{\bm{\alpha}}(\mathbf{x}) d \xi(\theta)$ normalizes the pseudo posterior distribution and\\ $r_{n,\bm{\alpha}}\left(\theta,\theta^{\ast}\right) = \mathop{\sum}_{i=1}^{n}\alpha_{i}\log\left\{p_{\theta_{i}^{\ast}}(x_{i}) / p_{\theta_{i}}(x_{i})\right\}$, which is a generalization of the definition from \citet{BFP2019AS} that uses a fixed, scalar weight to now incorporate risk-adjusted, record-indexed weights, $(\alpha_{i})_{i=1, \cdots, n}$ where each $\alpha_{i}(\mathbf{x})$ depends on the closely-held data.

We formulate the $\bm{\alpha}-$weighted log-pseudo likelihood,
\begin{equation}
f^{\bm{\alpha}}_{\theta}(\mathbf{x}) = \mathop{\sum}_{i=1}^{n} \alpha_i(\mathbf{x}) \log p_{\theta}(x_i),
\end{equation}
that we use to construct a pseudo posterior mechanism. 

\subsection{Main Results}
Our task is to specify assumptions that guarantee our pseudo posterior mechanism achieves an $\epsilon-$DP guarantee. In particular, we extend \citet{Dimitrakakis:2017:DPB:3122009.3122020} to show a direct relationship between the Lipschitz bound for the pseudo likelihood and the resulting $\epsilon-$ DP guarantee where both are a function of the record-indexed vector of weights, $\bm{\alpha}$, specified by the data provider. We present a collection of related results in this section with all of the associated proofs in Appendix \ref{app:theory}.

\subsubsection{Link the Global Lipschitz Bound to the Global DP Guarantee}
In this section and corresponding sections in Appendix \ref{app:theory}, we use the explicit notation $\bm{\alpha}(\mathbf{x})$ to emphasize the dependence of the $\alpha_{i} \leq 1$ on the closely-held data, $\mathbf{x}$.
We begin by extending the definition of DP from \citet{Dimitrakakis:2017:DPB:3122009.3122020} to our $\bm{\alpha}-$weighted pseudo posterior mechanism.
\begin{definition}
(Differential Privacy under the Pseudo Posterior Mechanism)
\begin{equation}
\mathop{\sup}_{\mathbf{x}\in \mathcal{X}^{n},\mathbf{y}\in \mathcal{X}^{n-1}:\delta(\mathbf{x}, \mathbf{y}) = 1} \mathop{\sup}_{B \in \beta_{\Theta}} \frac{\xi^{\bm{\alpha}(\mathbf{x})}(B \mid \mathbf{x})}{\xi^{\bm{\alpha}(\mathbf{y})}(B \mid \mathbf{y})} \leq e^{\epsilon}, \nonumber
\end{equation}
\end{definition}
\noindent which limits the change in the pseudo posterior distribution over all sets, $B \in \beta_{\Theta}$ (i.e. $\beta_{\Theta}$ is the $\sigma-$algebra of measurable sets on $\Theta$), from the inclusion of a single record (under the leave-one-out (LOO) distance, $\delta(\mathbf{x}, \mathbf{y}) = 1$, such that $\mathbf{y}$ differs from $\mathbf{x}$ by the omission of a \emph{single} data record).  Although the pseudo posterior distribution mass assigned to $B$ depends on $\mathbf{x}$, the $\epsilon$ guarantee is defined as the supremum over all $\mathbf{x} \in \mathcal{X}^{n}$.

Our main assumption extends \citet{Dimitrakakis:2017:DPB:3122009.3122020} to bound the \emph{log-pseudo likelihood ratio}, uniformly, for all databases, $\mathbf{y}\in\mathcal{X}^{n-1}$ that are at a LOO distance (i.e. $\delta(\mathbf{x},\mathbf{y}) = 1)$, over all $\mathbf{x}\in\mathcal{X}^n$ and over all $\theta \in \Theta$.   The uniform bound defines a maximum sensitivity in the log-pseudo likelihood from the inclusion of a record.  Our intuition that the magnitude of this sensitivity for the log-pseudo likelihood ratio is directly tied to the resulting $\epsilon-$ DP guarantee of the pseudo posterior is confirmed in is confirmed in Theorem~\ref{th:dpresult} for pseudo posterior draws of $\theta$ and in Lemma~\ref{lm:postpred} for the subsequent generation of a synthetic
database from a draw of $\theta$. 

\begin{assumption}
\label{ass:lipschitz}
(Lipschitz continuity)

\noindent Fix some $\theta \in \Theta$ and define a collection of record indexed mappings $\bm{\alpha}(\cdot)$: $\left\{ \mathcal{X}_i \rightarrow [0,1]\right\}_{n}$ for records $i = 1, \ldots n$
%$\bm{\alpha}^*(\mathbf{z}) = \{ \alpha^*(z_i)\}_{i = 1, \cdots, n} \in (0,1]^{n}$
and construct the Lipschitz function of $\theta$ over the space of databases,
\begin{eqnarray}\nonumber
\ell^{\bm{\alpha}}(\theta) &\overset{\Delta}{=}& \inf \left\{w: \abs{f^{\bm{\alpha}(\mathbf{x})}_{\theta}(\mathbf{x}) - f^{\bm{\alpha}(\mathbf{y})}_{\theta}(\mathbf{y})} \leq w, \forall \mathbf{x} \in \mathcal{X}^{n}, \mathbf{y}  \in \mathcal{X}^{n-1}: \delta(\mathbf{x}, \mathbf{y}) = 1\right\}.
\end{eqnarray}

\noindent Assumption~\ref{ass:lipschitz} restricts $\Theta$ such that the Lipschitz function of $\theta$ is uniformly bounded from above,
\begin{equation}\nonumber
\ell^{\bm{\alpha}}(\theta) \leq \Delta_{\bm{\alpha}} = \mathop{\sup}_{\theta\in\Theta}\left\{\ell^{\bm{\alpha}}(\theta)\right\}.
\end{equation}
\end{assumption}
\noindent Since $\bm{\alpha}(\cdot)$ is a vector of record-indexed functions, $\bm{\alpha}(\mathbf{x})$ and $\bm{\alpha}(\mathbf{y})$ only differ for a single record $j$ when $\mathbf{x}$ and $\mathbf{y}$ only differ in one record. Then $\bm{\alpha}(\mathbf{y}) =  \bm{\alpha}(\mathbf{x}_{-j})$. We note that the subscripting of $\Delta$ with $\bm{\alpha}$ is a notational device that denotes a Lipschitz bound computed using the log-pseudo likelihood, $f^{\bm{\alpha}(\mathbf{x})}_{\theta}(\mathbf{x})$ as contrasted with $\Delta$ computed using the unweighted posterior mechanism.  We further note that our general result simplifies to that of \citet{Dimitrakakis:2017:DPB:3122009.3122020} by specifying $\bm{\alpha} = \mathbf{1} $ and therefore $\ell^{\bm{\alpha}}(\theta)  \leq \Delta_{\bm{\alpha}} \leq\Delta$. 

We refer to $\Delta_{\bm{\alpha}}$ as ``global" over the space of databases, $\mathbf{x}\in\mathcal{X}^{n}$ and it represents the sensitivity of the $\bm{\alpha}-$weighted pseudo likelihood of Equation~(\ref{pseudolike}) that we use as our utility function. The Lipschitz function of $\theta$ and $\bm{\alpha}$, $\ell^{\bm{\alpha}}(\theta)$, is constructed using the pseudo log-likelihood, $f^{\bm{\alpha}(\mathbf{x})}_{\theta}(\mathbf{x})$ that incorporates record-indexed weights, $\bm{\alpha}(\mathbf{x})$, each of which is $\leq 1$.  Selecting an $\alpha_{i}$ close to zero indicates strong downweighting of a highly sensitive record for an unweighted posterior mechanism (with a high magnitude log-likelihood ratio for some $\theta \in \Theta$), which will reduce the sensitivity of that record under our pseudo posterior mechanism.  We see in our first two results that reducing the sensitivity of the log-likelihood ratio directly improves (i.e. reduces the value of $\epsilon$) the $\epsilon-$DP guarantee. 

Our next result directly connects the (global) Lipschitz bound, $\Delta_{\bm{\alpha}}$, for the log-pseudo likelihood of Assumption~\ref{ass:lipschitz} to resulting DP guarantee, $\epsilon = 2\Delta_{\bm{\alpha}}$, for each draw of $\theta$ from the pseudo posterior distribution. 

\begin{theorem}
\label{th:dpresult}
$\forall \mathbf{x} \in \mathcal{X}^{n}, \mathbf{y} \in \mathcal{X}^{n-1}:\delta(\mathbf{x}, \mathbf{y}) = 1, B \in \beta_{\Theta}$ (where $\beta_{\Theta}$ is the $\sigma-$algebra of measurable sets on $\Theta$) under $\bm{\alpha}(\cdot)$ with $\Delta_{\bm{\alpha}} > 0$ satisfying Assumption \ref{ass:lipschitz}:
\begin{equation}
\mathop{\sup}_{B \in \beta_{\Theta}} \,\,\, \mathop{\sup}_{\mathbf{x} \in \mathcal{X}^{n}, \mathbf{y}\in \mathcal{X}^{n-1}: \delta(\mathbf{x}, \mathbf{y}) = 1} \frac{\xi^{\bm{\alpha}(\mathbf{x})}(B \mid \mathbf{x})}{\xi^{\bm{\alpha}(\mathbf{y})}(B \mid \mathbf{y})} \leq \exp(2\Delta_{\bm{\alpha}}),
\end{equation}
i.e. the pseudo posterior $\xi^{\bm{\alpha}(\mathbf{x})}(\cdot \mid \mathbf{x})$ is $2\Delta_{\bm{\alpha}}-$DP.
\end{theorem}

This result directly connects the global Lipschitz bound to the global DP guarantee and will allow us to control the DP guarantee, indirectly, by setting the record-indexed weights, $\bm{\alpha}(\mathbf{x}) = \left(\alpha_{1}(\mathbf{x}),\ldots,\alpha_{n}(\mathbf{x})\right)$, that determines the Lipschitz bound.

Our next result extends our DP guarantee from pseudo posterior draws of $\theta$ for models that satisfy Assumption~\ref{ass:lipschitz} to draws of synthetic data, $\bm{\zeta} = (\zeta_{1},\ldots,\zeta_{m})$, constructed from the model pseudo posterior predictive distribution.  The generation of synthetic data is the purpose for the pseudo posterior mechanism.
\begin{lemma}
\label{lm:postpred}
Define $P^{\bm{\alpha}(\mathbf{x})}(\bm{\zeta} \in C \mid \mathbf{x}) = \int P(\bm{\zeta} \in C \mid \theta, \mathbf{x}) d \xi^{\bm{\alpha}(\mathbf{x})}(\theta \mid \mathbf{x})$ as the pseudo posterior predictive probability mass for $\bm{\zeta}$ in set $C \in \mathcal{A}^{n}$ (the $\sigma-$algebra of sets for $\mathcal{X}^{n}$), constructed from our pseudo posterior model for $\theta$ that satisfies DP with expenditure, $\epsilon$. Let $\bm{\zeta} = (\zeta_1, \ldots, \zeta_m)$ be $m$ independent draws from $P^{\bm{\alpha}(\mathbf{x})}(\bm{\zeta} \in C \mid \mathbf{x})$. This defines a mechanism for $\bm{\zeta}$ that satisfies DP with expenditure $\epsilon$ for any $m \leq n$.
\end{lemma}

We next formalize the method to construct our weighting scheme that characterizes our pseudo posterior mechanism.

\begin{assumption}
\label{ass:riskweight}
(Risk-based Weighting for Pseudo Posterior Mechanism)

\noindent Fix a value for $n$, the number of data records.  Let $m(\cdot)$ be a monotonically decreasing scalar function $m: [0,\infty) \rightarrow [0,1]$ such that $m(0) = 1$, and $m(\infty) = 0$. For every $\mathbf{x} \in \mathcal{X}^{n}$ choose a mapping $\bm{\alpha}(\cdot)$ such that
\begin{equation}
\label{setalpha}
\alpha_{i} = m\left(\mathop{\sup}_{\theta \in \Theta} \lvert f_{\theta}\left(x_{i}\right)\rvert\right) ,
\end{equation}
where $f_{\theta}\left(x_{i}\right)$ is computed from the unweighted, non-differentially private posterior synthesizer.  Under this procedure for selecting risk-based weights,
$\alpha_{i},~ i = 1,\ldots,n$, if $f_{\theta}\left(x_{i}\right)$ is non-finite for any $x_{i}$ and value of $\theta \in \Theta$, $\alpha_{i}$ is set to $m(\infty) =0$, which removes the contribution of database record, $i$, from the pseudo likelihood
of Equation~(\ref{pseudolike}) used to formulate the pseudo posterior mechanism of Equation~(\ref{pseudopost}).
\end{assumption}

The mapping $m(\cdot)$ in Assumption~\ref{ass:riskweight} includes threshold ($m(z) =  \mathbf{1}_{\{z < z^*\}}$) as well as smooth functions ($m(z) =  (z + 1)^{-1}$), providing us the flexibility for how to implement the weighting in practice. Since we remove the likelihood contributions for all database records with non-finite log-likelihoods by setting their associated weights in our pseudo posterior mechanism to $m(\infty) = 0$, our mechanism is guaranteed to satisfy Assumption~\ref{ass:lipschitz} with a finite $\Delta_{\bm{\alpha}} < \infty$ and thus be globally differentially private. This is a non-asymptotic result at every $n$; however we want to \emph{estimate} the global $\Delta_{\bm{\alpha}}$ (and, therefore, $\epsilon$), rather than simply knowing it exists.

\begin{comment}
Lemma~\ref{le:llbound} demonstrates that overall Lipschitz is bounded above by the maximum of by-record Lipschitz bounds, each computed as the supremum over the parameter space, $\Theta$, of the absolute value of the log-pseudo likelihood computed for datum, $x_{i}$.  It is straightforward to conclude that the disclosure risk for each record, $i$, is proportional to this quantity such that it is natural to set our record-indexed weights, $\alpha_{i}, i \in (1, \cdots, n)$, to be inversely proportional to these by-record bounds.
\end{comment}

We use Assumption~\ref{ass:riskweight} to implement our $\bm{\alpha}-$weighted pseudo posterior mechanism.  Fix a database, $\mathbf{x}$, and compute a record-indexed vector of log-likelihood ratios, $\vert f_{\theta}\left(x_{i}\right)\vert$ and linearly transform them to $\vert \tilde{f}_{\theta,i}\vert \in [0,1]$ such that records with lower values for $\vert f_{\theta,i} \vert$, that indicate lower identification risks, produce values of $\tilde{f}_{\theta,i}$ near $0$. We, next, set $\alpha_i =  c \times (1-\tilde{f}_{i}) + g$ where $c$ and $g$ may be used by the data provider to scale and shift the weights, respectively, restricted to $\alpha_{i} \in [0,1],~\forall i \in(1,\ldots,n)$ in order to achieve a desired  Lipschitz bound, $\Delta_{\mathbf{\alpha},\mathbf{x}}$, for database, $\mathbf{x}$ and the local DP privacy guarantee of $\epsilon_{\mathbf{x}} = 2\Delta_{\mathbf{\alpha},\mathbf{x}}$ (\citet{HuSavitskyWilliams2021rebds} demonstrate the uses of $c$ and $g$ to fine tune the risk-utility trade-off in non-differerentially private synthetic data settings). So the data provider indirectly controls the local privacy guarantee by formulating the weights.  We discuss an asymptotic method in Section~\ref{sec:lockin} that ``discovers" a global Lipschitz bound and associated global $\epsilon$ of an $(\epsilon,\pi)-$ probabilistic DP guarantee from a local result.  We show that the $\pi$, the probability of deviating from $\epsilon-$ DP, contracts onto $0$ for a sufficiently large sample size, $n$. 

\begin{comment}
To the extent that a given local database, $\mathbf{x}^{'}$, contains relatively many records, $i$, such that Equation~(\ref{setalpha}) is non-finite, more of the likelihood contributions for those records will be removed from the computation of the pseudo posterior in Equation~(\ref{pseudopost}), with the result that prior smoothing will induce more distortion in the resulting synthetic data. This greater degree of smoothing will, in turn, reduce the utility of the synthetic dataset and also increase the privacy protection, so the local $\Delta_{\bm{\alpha}, \mathbf{x}^{'}}$ will be relatively small.   By contrast, to the extent a local database, $\mathbf{x}^{''}$, contains few-to-no non-finite likelihood log-likelihood values, the resulting local Lipschitz bound, $\Delta_{\bm{\alpha},\mathbf{x}^{''}}$, will be larger since the risk and utility would be relatively higher such that the local $\Delta_{\bm{\alpha},\mathbf{x}^{''}}$ would be closer to the global $\Delta
_{\bm{\alpha}}$.
\end{comment}

In our application to a local data set $\mathbf{x}$ we may want to use a weighting scheme $\bm{\alpha}(\mathbf{x})$ which mildly violates the stated conditions in Assumption \ref{ass:lipschitz}. In particular, we consider estimation of $\alpha_i$ which weakly depends on $x_j$ for $i\ne j$, where this dependence attenuates asymptotically. For example, we use estimates of $\theta$ from an unweighted posterior distribution which weakly depend on all values of $x_i$. We see asymptotically that this dependence decays as $\theta$ collapses to a point, such that the results in this section apply to this case of weak dependence among the $\alpha_{i}$ except for minor updates to notation. %See Appendix \ref{app:approxalpha} for more details.

\section{Computing a Local Lipschitz Bound}\label{sec:setweights}
In this section, we describe the implementation algorithm to compute the pseudo likelihood weights, $\bm{\alpha} = (\alpha_{1},\ldots,\alpha_{n})$ for a local database, $\mathbf{x}$, from the \emph{unweighted} synthesizer and the subsequent computation of the local Lipschitz bound, $\Delta_{\bm{\alpha},\mathbf{x}}$, associated with the pseudo posterior mechanism. In Section \ref{sec:lipschitz:ppEM}, we lay out the connection between the scalar-weighted pseudo posterior mechanism and the EM, with a discussion of the implications on the data utility of locally differentially private synthetic data generated under the two mechanisms.

\label{sec:lipschitz}

\begin{enumerate}
\item Compute weights $\bm{\alpha}$
	\begin{enumerate}
	\item Let $\lvert f_{\theta_{s},i} \rvert$ denote the absolute value of the log-likelihood computed from the unweighted pseudo posterior synthesizer for database record, $i \in (1,\ldots,n)$ and MCMC draw, $s \in (1,\ldots,S)$ of $\theta$.
	\item Compute the $S \times n$ matrix of by-record (absolute value of) log-likelihoods, $L = \left\{ \lvert f_{\theta_{s},i} \rvert \right\}_{i=1,\ldots,n,~s=1,\ldots,S}$.
\item Compute the maximum over each $S\times 1$ column of $L$ to produce the $n\times 1$ (database record-indexed) vector, $\mathbf{f} = \left(f_{1},\ldots,f_{n}\right)$.  We use a linear transformation of each $f_{i}$ to $\tilde{f}_{i} \in [0,1]$ where values of $\tilde{f}_{i}$ closer to $1$ indicates relatively higher identification disclosure risk:
$\tilde{f}_{i} = \frac{f_i - \min_j f_j}{\max_j f_j - \min_j f_j}$.
    \item We formulate by-record weights, $\bm{\alpha} = (\alpha_1, \cdots, \alpha_n)$,
    \begin{equation}
    \alpha_i =  c \times (1-\tilde{f}_{i}) + g,
    \label{eq:vectroweights}
    \end{equation}
    where $c$ and $g$ denote a scaling and a shift parameters, respectively, of the $\alpha_i$ used to tune the risk-utility trade-off.  If we set scaling tuning parameter, $c=1$ and shift tuning parameter, $g=0$, then each $\alpha_{i}$ is simply $(1-\tilde{f}_{i})$ such that the pseudo likelihood weights are solely a function of the record-indexed log likelihoods.  As discussed in \citet{HuSavitskyWilliams2021rebds}, decreasing $c < 1$ will compress the distribution of the $(\alpha_{i})$ while setting $g < 0$ will shift downward the distribution of the weights such that more weights will be close to $0$.  We use truncation to ensure each $\alpha_{i} \in [0,1]$.  These $\bm{\alpha}$ satisfy a slightly weaker asymptotic form of Assumptions ~\ref{ass:lipschitz} and \ref{ass:riskweight}. % See Appendix \ref{app:approxalpha} for more detail.

\begin{comment}
As discussed in \citet{HuSavitskyWilliams2021rebds}, the scaling parameter $c$ compresses or expands the distribution by-record weights and induces a global affect on the risk-utility trade-off, while the shift parameter $g$ shifts the distribution of by-record weights has a local effect on the risk-utility trade-off. 
\end{comment}

We will show in Section \ref{sec:app} the effects of different configurations of $c$ and $g$ on the risk and utility profiles of the differentially private synthetic dataset for the CE sample, generated under our proposed $\bm{\alpha}-$weighted pseudo posterior mechanism.
	\end{enumerate}
\item Compute Lipschitz bound, $\Delta_{\bm{\alpha},\mathbf{x}}$
	\begin{enumerate}
	\item Use $\bm{\alpha} = \left(\alpha_{1},\ldots,\alpha_{n}\right)$ to construct the pseudo likelihood of Equation~\ref{pseudolike} from which the pseudo posterior of Equation~\ref{pseudopost} is estimated.  Draw $(\theta_{s})_{s=1,\ldots S}$ from the $\bm{\alpha}-$weighted pseudo posterior distribution.
	\item As earlier, compute the $S\times~n$ matrix of log-pseudo likelihood values, $L^{\bm{\alpha}} = \left\{ \lvert f^{\bm{\alpha}}_{\theta_{s},i} \rvert \right\}_{i=1,\ldots,n,~s=1,\ldots,S}$
	\item Compute $\Delta_{\bm{\alpha},\mathbf{x}} = \mathop{\max}_{s,i} \lvert f^{\bm{\alpha}}_{\theta_{s},i} \rvert$.
	\end{enumerate}
\item Draw synthetic data, $\bm{\zeta}_{\ell}$, from the pseudo posterior distribution
	\begin{enumerate}
	\item Using the $(\theta_{s})_{s=1,\ldots S}$ drawn from the $\bm{\alpha}-$weighted pseudo posterior distribution estimated in the earlier step, randomly sample $\ell = 1,\ldots,(m = 20)$ parameter values and draw synthetic data value, $\zeta_{\ell,i} \ind p_{\theta_{\ell}}(\cdot)$ for parameter draw $\ell \in (1,\ldots,m)$ and database record $i \in (1,\ldots,n)$.  This step accomplishes a draw from the pseudo posterior predictive distribution.
	\item Release the synthetic data, $\bm{\zeta} = (\bm{\zeta}_1, \cdots, \bm{\zeta}_m)$, in place of the closely-held real data, $\mathbf{x}$.
	\end{enumerate}
\end{enumerate}

Our pseudo posterior mechanism \emph{indirectly} sets the local DP guarantee, $2\Delta_{\mathbf{\alpha},\mathbf{x}}$ through the computation and subsequent scaling and shifting of the likelihood weights, $\bm{\alpha}$.  

\subsection{Exponential Mechanism Reduces to Scalar Weighting}
\label{sec:lipschitz:ppEM}

\citet{WassermanZhou2010, ZhangRubinsteinDimitrakakis2016AAAI, SnokeSlavkovic2018PSD} use the EM to generate synthetic data with privacy guarantees from a non-private mechanism. Suppose we start with a non-private mechanism, such as an unweighted posterior synthesizer,
\begin{equation}
\xi\left(\theta \mid \mathbf{x},\gamma\right) \propto \left[\mathop{\prod}_{i=1}^{n}p\left(x_{i} \mid \theta\right)\right]\xi\left(\theta\mid \gamma\right).
\label{eq:unweighted}
\end{equation}

Under the set-up of \citet{ZhangRubinsteinDimitrakakis2016AAAI} that uses the log-likelihood function as the utility function, i.e. $u(\mathbf{x}, \theta) = \log\left(\mathop{\prod}_{i=1}^{n}p\left(x_{i} \mid \theta\right)\right)$, the EM generates private samples from
\begin{equation}
\hat{\theta} \propto \exp \left(\frac{\epsilon \, \log\left(\mathop{\prod}_{i=1}^{n}p\left(x_{i} \mid \theta\right)\right)}{2\Delta}\right) \xi\left(\theta\mid \gamma\right),
\end{equation}
where the prior, $\xi(\theta\mid \gamma)$, is chosen as the ``base" distribution as specified by \citet{McSherryTalwar2007} that ensures the EM produces a proper density function.  Furthermore,
\begin{eqnarray}
\exp \left(\frac{\epsilon \, \log\left(\mathop{\prod}_{i=1}^{n}p\left(x_{i} \mid \theta\right)\right)}{2\Delta}\right) \xi\left(\theta\mid \gamma\right)
&=& \exp (\log (\mathop{\prod}_{i=1}^{n}p\left(x_{i} \mid \theta\right))^{\frac{\epsilon}{2\Delta}}) \xi\left(\theta\mid \gamma\right) \nonumber \\
&=& \left(\mathop{\prod}_{i=1}^{n}p\left(x_{i} \mid \theta\right)^{\frac{\epsilon}{2\Delta}}\right) \xi\left(\theta\mid \gamma\right),
\label{eq:EMpp}
\end{eqnarray}
which demonstrates that the EM under a log-likelihood utility is equivalent to a risk-adjusted, scalar-weighted pseudo posterior synthesizer with scalar weight $\frac{\epsilon}{2\Delta}$, where $\alpha_i = \frac{\epsilon}{2\Delta},~\forall i \in (1,\ldots,n)$.  \citet{wang2015privacy} derived this same scalar-weighted result in their implementation of a gradient descent algorithm to sample the EM under a pseudo log-likelihood utility.

Using a scalar weight, $\alpha_i = \frac{\epsilon}{2\Delta},~\forall i \in (1,\ldots,n)$, shown in Equation (\ref{eq:EMpp}), we expect a resulting lower utility for synthetic data draws under this mechanism than we do under our $\bm{\alpha}-$weighted pseudo posterior shown in Equation (\ref{eq:pp}), which uses a vector of record-indexed weights. The $\bm{\alpha}-$weighted pseudo posterior is more surgical and concentrates the downweighting to records with higher risk, whereas the EM must downweight all records the same amount.  Downweighting all records the same amount will be conservative because the scalar weight is based on the \emph{worst case} sensitivity, $\Delta$, over the entire database of records and the parameter space, which is required to achieve a local DP privacy guarantee,  and not tuned to the risk ($\tilde{f}_{i}$) of each record.

The re-casting of the EM as a scalar-weighted pseudo likelihood under a log-likelihood utility also provides insight into why our $\bm{\alpha}-$weighted pseudo posterior mechanism sets the $\epsilon-$ DP guarantee indirectly through specification of the vector of weights, $\bm{\alpha} = (\alpha_{1},\ldots,\alpha_{n})$, that determines $\Delta_{\bm{\alpha}}$, which in turn, determines $\epsilon = 2\Delta_{\bm{\alpha}}$.  Since the commonly-used EM utilizes a single, scalar weight for all records, it is straightforward to directly set $\epsilon$, but at a tremendous loss of efficiency in terms of risk-utility trade-off as compared to the $\bm{\alpha}-$weighted pseudo posterior mechanism.  So our mechanism achieves a higher utility for an equivalent guarantee, $\epsilon$. 

We illustrate in Section \ref{sec:app} the reduction in utility of the local differentially private synthetic dataset generated under the EM, compared to that generated under our proposed $\bm{\alpha}-$weighted pseudo posterior mechanism, at an equivalent privacy guarantee for both mechanisms.

\section{Turning A Local Bound into A Global Bound}\label{sec:lockin}

In this section we proceed to demonstrate that a local Lipschitz bound or sensitivity, $\Delta_{\bm{\alpha},\mathbf{x}}$, computed on observed database, $\mathbf{x}$, contracts on or becomes arbitrarily close to $\Delta_{\bm{\alpha}}$, the global Lipschitz bound or supremum over the space of databases, $\mathcal{X}^{n}$, for sample size, $n$, sufficiently large.

\subsection{Asymptotic Convergence of Local Lipschitz to Global Lipschitz }\label{sec:asympt}

Although our DP result is non-asymptotic for every $n$, in the sense that we have earlier shown that a \emph{finite} global $\Delta_{\bm{\alpha}}$ is guaranteed to exist under our $\bm{\alpha}-$weighted pseudo posterior mechanism, we nevertheless do not \emph{know} its value.  We employ asymptotics to learn the global Lipschitz bound, $\Delta_{\bm{\alpha}}$, to any degree of desired precision.  We develop a contraction result for any $\bm{\alpha}-$weighted pseudo distribution to demonstrate under a set of conditions that convergence of the pseudo posterior distribution leads to asymptotic convergence of the local Lipschitz bound, $\Delta_{\bm{\alpha},\mathbf{x}}$, to the global bound, $\Delta_{\bm{\alpha}}$ in $P_{\theta_{0}}-$probability for $n$ sufficiently large.

Our asymptotic contraction of the local Lipschitz bound onto the global Lipschitz bound (that has a direct functional relationship to the global privacy guarantee, $\epsilon$) does not provide a global $\epsilon-$ DP \emph{guarantee} because there is the possibility of leakage of private information, $\pi$, at any fixed sample size such that our computed $\epsilon$ on a local database may be exceeded.  Therefore, we employ our asymptotic result on the contraction of Lipschitz bounds to claim an $(\epsilon,\pi)-$ probabilitistic DP guarantee where $\delta$ represents a probability that there are some databases in the space of databases for which $\epsilon$ is exceeded.  Under our asymptotic contraction of local Lipschitz bounds to the global bound, we achieve that $\pi$ contracts onto $0$.

We formally introduce a definition for probabilistic differential privacy (pDP) that adapts the formulation of  \citet{4497436} to our $\bm{\alpha}-$weighted pseudo posterior mechanism.

\begin{definition}
(Probabilistic Differential Privacy)
Let $\epsilon > 0$ and $ 0 < \pi < 1$.
We say that our pseudo posterior mechanism is $(\epsilon,\pi)$-probabilistically differentially private (pDP) if $\forall \mathbf{x} \in \mathcal{X}^{n}$,
\begin{equation}\label{eq:pDP}
 \mbox{Pr}\left(\mathbf{x} \in \mbox{Disc}(\mathbf{x},\epsilon)\right) \leq \pi, \nonumber
\end{equation}
\end{definition}
\noindent where the probability is taken over $\mathbf{x} \in \mathcal{X}^{n}$ and $\mbox{Disc}(\mathbf{x},\epsilon)$ denotes
the \emph{disclosure} set, \newline $\{\mathbf{x}\in\mathcal{X}^{n}:
\mathop{\sup}_{B \in \beta_{\Theta}} \log\left(\frac{\xi^{\bm{\alpha}(\mathbf{x})}(B \mid \mathbf{x})}{\xi^{\bm{\alpha}(\mathbf{y})}(B \mid \mathbf{y})}\right) > \epsilon ,~\forall \mathbf{y}:\delta(\mathbf{x}, \mathbf{y}) = 1\}$,
the subspace of $\mathcal{X}^{n}$ where our $\bm{\alpha}-$weighted pseudo posterior mechanism exceeds an $\epsilon-$DP guarantee.

This definition constructs a probability for the event that there are \emph{any} databases in the space of databases for which our pseudo posterior mechanism \emph{exceeds} $\epsilon$ under the leave-one-out (LOO) distance.
\begin{comment}
We reveal in the sequel an asymptotic result that the parameter space $\Theta$ contracts onto a point $\theta^{\ast}$ such that all $\mathbf{x} \in \mathcal{X}^{n}$ are, in turn, drawn from a single distribution indexed by $\theta^{\ast}$. This result produces the outcome that for sample size, $n$, sufficiently large, the local Lipschitz, $\Delta_{\bm{\alpha},\mathbf{x}}$ contracts onto $\Delta_{\bm{\alpha}}$, the global Lipschitz.
\end{comment}
We recall that our vector weights, $\bm{\alpha} = (\alpha_{1},\ldots,\alpha_{n})$ determine $\Delta_{\bm{\alpha}}$, which \emph{indirectly} sets $\epsilon \leq 2\Delta_{\bm{\alpha}}$.   Our asymptotic result on the contraction of the local to global Lipschitz bound, presented in this section, reveals that $\pi$, which represents the (maximum) probability that $\epsilon-$ DP is exceeded, limits to $0$ in $P_{\theta_{0}}-$ probability.

We verify our theoretical result by conducting a simulation study in Section \ref{asympt:sim} that demonstrates the contraction of the distribution for the local $\Delta_{\bm{\alpha},\mathbf{x}}$ for a relatively moderate sample sizes. Furthermore, we suggest a procedure for selecting a global $\epsilon$ that would result in a very small-to-negligible $\pi$.

\subsection{Preliminaries}
We next demonstrate the frequentist properties of our pseudo posterior Bayesian estimator.  We generalize the result of \citet{BFP2019AS} developed for a fixed, scalar weight to our vector of record-indexed weights that depend on the closely-held data.  Suppose $x_{1},\ldots,x_{n} \ind P_{\theta_{0}}$ for $\theta_{0} \in \Theta$. Under \emph{frequentist} consistency, the $\mathbf{x} = (x_{1},\ldots,x_{n})$ are random with respect to $P_{\theta_{0}}$ (for fixed $\theta_{0}$), so taking probabilities and expectations with respect to $P_{\theta_{0}}$ requires us to address the dependence of $\alpha_{i}$ on $\mathbf{x}$ to construct the contraction rate for correctness and thoroughness. We drop the notation denoting the explicit dependence of $\alpha_i(x_i)$ for exposition of our consistency results in the sequel and just use $\alpha_{i}$ for readability when the context is clear.

Since our pseudo posterior formulation induces misspecification, we allow the true generating parameters, $\theta_{0}$, to lie outside the parameter space, $\Theta$.   We will show in the sequel that our model contracts on $\theta^{\ast} \in \Theta$ in $P_{\theta_{0}}-$probability, where $\theta^{\ast}$ is the point that minimizes the Kullback-Liebler (KL) divergence from $P_{\theta_{0}}$; that is,
\begin{equation}\label{misMLE}
\theta^{\ast} := \mathop{\arg\min}_{\theta\in\Theta}D\left(p_{\theta},p_{\theta_{0}}\right),
\end{equation}
where $D(p,q) = \int p\log(p/q)d\mu$ for dominating measure, $\mu$.
%We later demonstrate that such weighting of likelihood contributions enhances privacy protection in the form of a reduced DP expenditure as compared to an unweighted posterior distribution mechanism.

Our asymptotic result on the contraction in $P_{\theta_{0}}-$probability relies on bounding the $\bm{\alpha}-$R\'{e}nyi divergence measure,
\begin{equation}
D^{(n)}_{\theta_{0},\bm{\alpha}}\left(\theta,\theta^{\ast}\right) = \mathop{\sum}_{i=1}^{n}D_{\theta_{0},\alpha,i}\left(\theta,\theta^{\ast}\right) = \mathop{\sum}_{i=1}^{n}\frac{1}{\alpha_{i}-1}\log\left\{A_{\theta_{0},\alpha,i}\left(\theta,\theta^{\ast}\right)\right\},
\end{equation}
where $A_{\theta_{0},\alpha,i}\left(\theta,\theta^{\ast}\right) = \mathop{\int} \left(\frac{p_{\theta_{i}}}{p_{\theta^{\ast}_{i}}}\right)^{\alpha_{i}}p_{\theta_{0},i}d\mu_{i}$ under dominating measure $\mu_{i}$
is defined as the $\bm{\alpha}-$affinity for record, $x_{i}$, such that $A^{(n)}_{\theta_{0},\bm{\alpha}}\left(\theta,\theta^{\ast}\right) = \mathop{\prod}_{i=1}^{n}A_{\theta_{0},\alpha,i}\left(\theta,\theta^{\ast}\right)$, the $\bm{\alpha}-$affinity for the product measure space, where we have updated definitions from use of a scalar, $\alpha$, to record-indexed $\alpha_{i}$.

The posterior probability of the $\bm{\alpha}-$R\'{e}nyi distance between $\theta\in\Theta$ and the point $\theta^{\ast}$ limits to $0$ at a rate that is a function of $n$ for any weighting scheme, $\bm{\alpha}(\mathbf{x})$, where the construction of $\bm{\alpha}$ depends on the observed data, $\mathbf{x}$, as does ours.  We require the following two conditions to achieve contraction of the local $\Delta_{\bm{\alpha},\mathbf{x}}$ to the global $\Delta_{\bm{\alpha}}$:

\begin{assumption}
\label{prior}
(Prior mass covering truth)
We construct a KL neighborhood of $\theta^{\ast}$ with radius, $\eta$,with,
\begin{multline}
B_{n}\left(\theta^{\ast},\eta;\theta_{0}\right) = \left\{\theta\in\Theta: \mathop{\sum}_{i=1}^{n}\int p_{\theta_{0},i}\log\left(p_{\theta^{\ast}_{i}} / p_{\theta_{i}}\right)d\mu_{i}\leq n\eta^{2},\right.\\
 \left.\mathop{\sum}_{i=1}^{n}\int p_{\theta_{0},i}\log^{2}\left(p_{\theta^{\ast}_{i}} / p_{\theta_{i}}\right)d\mu_{i}\leq n\eta^{2}~\right\}.
\end{multline}
Restrict the prior, $\xi$, to place positive probability on this KL neighborhood,
\begin{equation}\label{smallballprior}
  \xi\left(B_{n}\left(\theta^{\ast},\eta;\theta_{0}\right)\right) \geq e^{-n\tau_{n}^{2}}.
\end{equation}
\end{assumption}

\begin{assumption}
\label{size}
(Control size of $\bm{\alpha}$)
Let $A_{n} := \displaystyle\left\{i: \alpha_{i} < 1^{-}; i \in 1,\ldots,n\right\}$ and $n_{A} := \vert A_{n} \vert$, where $\vert A_{n} \vert$ denotes the number of elements in $A_{n}$.
Let \\ $Q_{n} := \displaystyle\left\{i: \alpha_{i} = \alpha^{(n)} \geq 1^{-}; i \in 1,\ldots,n\right\}$ for some constant $\alpha^{(n)}$ and $n_{Q} := \vert Q_{n} \vert$.
\begin{align*}
&\displaystyle\lim\mathop{\sup}_{n}\left\lvert A_{n}\right\rvert = \lim\mathop{\sup}_{n} n_{A} = \mathcal{O}\left(n^{\frac{1}{2}}\right), \text{ with $P_{\theta_{0}}-$probability $1$},\\
&\lim\mathop{\sup}_{n} (1-\alpha^{(n)}) = \mathcal{O}\left(n_{Q}^{-\frac{1}{2}}\right), \text{ with $P_{\theta_{0}}-$probability $1$},
\end{align*}
such that for constants $C_{1}, C_{3} > 0$ and $n$ sufficiently large,
\begin{align*}
&\displaystyle\mathop{\sup}_{n}\vert A_{n} \vert \leq C_{1}n^{\frac{1}{2}},\\
&\displaystyle\mathop{\sup}_{n}(1-\alpha^{(n)}) \leq C_{3}\tau_{n}n_{Q}^{-\frac{1}{2}}.
\end{align*}
\end{assumption}
These two assumptions are required for consistency of our $\bm{\alpha}-$pseudo posterior mechanism at $\theta^{\ast}$.  The first assumption requires the prior to place some mass on a KL ball near $\theta^{\ast}$ as defined in Equation~(\ref{misMLE}). The second assumption outlines a dyadic subgrouping of data records, where $A_{n}$ contains those records whose likelihood contributions are downweighted to lessen the estimated identification disclosure risk (and improve privacy) for those records in the resulting synthetic data.  The second subset of records, $Q_{n}$, contains those records that are minimally downweighted due to nearly zero values for identification disclosure risks.  Since $\alpha_{i} \leq 1,~\forall i\in (1,\ldots,n)$, the constant value, $\alpha^{(n)}$, for all units in $Q_{n}$ approaches $1$ from the left.  We show that the consistency result to $\theta^{\ast}$ for the synthesizer is dominated by the likelihood weighting for records in the downweighted set, $A_{n}$.  Assumption~\ref{size} restricts the number of downweighted records (where $\alpha_{i} < 1^{-}$) to grow at a slower rate than the sample size, $n$, such that the downweighting becomes relatively more sparse.  Our experience demonstrates that when weights are constructed based on disclosure risks, downweighting is confined to isolated records, which are sparse.

\begin{comment}
We set each $\alpha_{i}(\mathbf{x})$ on the set $A_{n}$ based on the actual data value for observed record, $i$, $x_{i}$, and the synthetic data $x^{\ast}_{1},\ldots,x^{\ast}_{n} \sim p_{T}(x^{\ast}|\mathbf{x})$, with implicit conditioning on the model, $T$, after integrating out $\theta$ from the synthesizer.  So the $(\alpha_{i})\in A_{n}$ may be expected to express mutual dependence in the general case for assessing frequentist consistency, unlike the $(x_{i})$, which are assumed to be independent.  While our consistency result allows for dependence among the $(\alpha_{i})\in A_{n}$, Assumption~\ref{size} restricts the number of downweighted records (where $\alpha_{i} < 1^{-}$) to grow at a slower rate than the sample size, $n$, such that the downweighting becomes relatively more sparse.
\end{comment}

\begin{theorem}
\label{th:consistency}
(Contraction of the $\bm{\alpha}-$pseudo posterior distribution).
\newline Let $\displaystyle\bm{\alpha} = \left(\alpha_{1} \in [0,1],\ldots,\alpha_{n}\in[0,1]\right)$.  Define $\displaystyle\alpha_{m} := \max_{i\in A_{n}}\alpha_{i} \in [0,1]$ and $\displaystyle\alpha_{l} := \min_{i\in A_{n}}\alpha_{i} \in [0,1]$.  Let $D^{(n_{A})}_{\theta_{0},\bm{\alpha}}\left(\theta,\theta^{\ast}\right) = \sum_{i\in A_{n}}D_{\theta_{0},\alpha,i}$ and $D^{(n_{Q})}_{\theta_{0},1^{-}}\left(\theta,\theta^{\ast}\right) = \sum_{i\in Q_{n}}D_{\theta_{0},1^{-},i}$.   Let $\theta^{\ast}$ be as defined in Equation~(\ref{misMLE}).  Assume that $\tau_{n}$ satisfies $n\tau_{n}^{2} \geq 2$ and suppose Assumptions~\ref{prior} and \ref{size} hold.  Let $C^{\ast}_{1} = \sqrt{2+C_{1}^{2}+C_{3}^{2}} \geq \sqrt{2}$. Then for any $D \geq 2$ and $t > 0$,
\begin{equation}\label{postcontract}
\xi^{\bm{\alpha}}\left(\frac{1}{n}\left[(1-\alpha_{m})D^{(n_{A})}_{\theta_{0},\bm{\alpha}}\left(\theta,\theta^{\ast}\right) + (1-\alpha^{(n)})D^{(n_{Q})}_{\theta_{0},1^{-}}\left(\theta,\theta^{\ast}\right) \right]\geq (D+3t)\tau_{n}^{2}\big\vert \mathbf{x}\right) \leq e^{-t n\tau_{n}^{2}},
\end{equation}
hold with $P_{\theta_{0}}-$probability at least $1-\left[(\alpha_{l}^{2}+2) (C^{\ast}_{1})^{2} / \alpha_{m}^{2} \times 2/\left\{(D+t-1)^{2}n\tau_{n}^{2}\right\}\right]$.
\end{theorem}
Since $(1-\alpha^{(n)}) = \mathcal{O}(n_{Q}^{-1/2})$, while $n_{A} = \mathcal{O}(n^{1/2})$, the first term dominates with increasing $n$, so that the $(1-\alpha_{m})^{-1}$ is the dominating penalty on the $\tau_{n}$ contraction rate of the $\bm{\alpha}-$pseudo posterior onto $\theta^{\ast}$.   Even though the downweighting becomes relatively more sparse due to Assumption~\ref{size}, it is the maximum value of $\alpha_{i}$ for $i \in A_{n}$ on the set of downweighted records that penalizes the rate. We observe that the rate of contraction is injured by factor, $(1-\alpha_{m})^{-1}$.  Since $\alpha_{i} \leq 1^{-},~\forall i \in A_{n}$, our result generalizes \citet{BFP2019AS} to allow a tempering of a \emph{portion} of the posterior distribution and there is a penalty to be paid in terms of contraction rate for the tempering.
\begin{comment}
The rate in $P_{\theta_{0}}-$probability is further injured by a factor, $(2+\alpha_{l}^{2})/\alpha_{m}^{2}\times (C^{\ast}_{1})^{2}$, where $C^{\ast}_{1} \geq 2$, relative to \citet{BFP2019AS}.  This additional penalty arises because vector-weighting shifts posterior mass from high to low risk regions of the data support, which induces a distribution contraction.
\end{comment}
Since we induce the misspecification through the weights, $\bm{\alpha}$, the distance of the point of contraction, $\theta^{\ast}$ from the true generating parameters, $\theta_{0}$, and the contraction rate on this point are \emph{both} impacted by the induced misspecification.  The requirement for increasing sparsity in the number of downweighted record likelihood contributions, however, ensures that $\theta^{\ast}$ will be relatively close to $\theta_{0}$ that produces a high utility for our (pseudo posterior) estimator.

If we plug in for $\tau_{n}$, we see that our contraction of $\Theta$ to $\theta^{\ast}$ occurs at a rate that is of $\mathcal{O}(n^{-1/2})$.

\subsection{Contraction of Local Lipschitz bound onto Global bound}
Asymptotically Theorem~\ref{th:consistency} guarantees that the space $\theta\in\Theta$ collapses onto $\theta^{\ast}$ for $n$ sufficiently large.  The space of databases, $\mathcal{X}^{n}$, drawn under this distribution collapses unto a single distribution, $\tilde{x} \ind P^{\bm{\alpha}}(x|\theta^{\ast})$ with density  $p^{\bm{\alpha}}(x|\theta^{\ast}) \propto \exp(f^{\bm{\alpha}}_{\theta^{\ast}}(x))$. The term $\tilde{x}$ denotes the risk-corrected version of $\mathbf{x}$ under which high disclosure risk records are less likely to be drawn due to their downweighting. 
%The support $\tilde{x} \in \tilde{\mathcal{X}} \subset{\mathcal{X}}$ is a truncated subset of the original support. 
High-risk records are isolated relative to other records located in regions of the unweighted generating distribution, $P_{\theta_{0}}$, of low-probability mass, such as the tails.
%The records $x \notin \tilde{\mathcal{X}}$ receive an $\alpha = 0$ weight and thus are ignored in the inference.
Since the contraction of the pseudo posterior distribution induces the collapsing of the parameter space to a point and the space of databases to a single distribution (conditioned on $\theta^{\ast}$) for large $n$, this result guarantees that the local Lipschitz bound, $\Delta_{\bm{\alpha},\mathbf{x}}$ and the the global bound $\Delta_{\bm{\alpha}}$ contract together for $n$ sufficiently large. First, we revisit the different forms of the $\Delta$ bound:
\begin{comment}
\begin{align}
 \Delta_{\bm{\alpha}} &= \sup_{\theta \in \Thet} \left\{\inf \left\{w: \abs{f^{\bm{\alpha}(x_i)}_{\theta}(x_i)} \leq w, \forall x_i \in \mathbf{x},~\forall \mathbf{x} \in \mathcal{X}^{n} \right\} \right\}\\
  \Delta_{\bm{\alpha},\mathbf{x}} &= \max_{\theta_s \sim \xi^{\alpha}(\theta|\mathbf{x})} \left\{\inf \left\{w: \abs{f^{\bm{\alpha}(x_i)}_{\theta}(x_i)} \leq w, i \in \{1, \ldots n\}, x_i \sim P(x|\theta_0) \right\} \right\}\\
   \Delta^{\infty}_{\bm{\alpha},\mathbf{x}} & = \inf \left\{w: \abs{f^{\bm{\alpha}(x)}_{\theta^{\ast}}(x)} \leq w, \forall x \in \tilde{\mathcal{X}} \right\}
\end{align}
\end{comment}
\begin{align}
\Delta_{\bm{\alpha}} &= \sup_{\theta \in\Theta} \left\{\inf \left\{w: \abs{f^{\bm{\alpha}(x_i)}_{\theta}(x_i)} \leq w, \forall x_i \in \mathbf{x},~\forall \mathbf{x} \in \mathcal{X}^{n} \right\} \right\}\label{wholespace}\\
 &= \sup_{\theta \sim \xi^{\alpha}(\theta|\mathbf{x})} \left\{\inf \left\{w: \abs{f^{\bm{\alpha}(x_i)}_{\theta}(x_i)} \leq w, \forall x_i \in \mathbf{x},~\forall \mathbf{x} \in \mathcal{X}^{n} \right\} \right\}\label{global}\\
\end{align}

The formulation of $\Delta_{\bm{\alpha}}$ comes directly from Assumption \ref{ass:lipschitz}. This is a bound of the $\alpha$-weighted log-likelihood over the full support of both $\theta$ and $\mathbf{x}$. Note that we can replace the $\sup$ over $\Theta$ in Equation~\ref{wholespace} with the $\sup$ over infinite draws from the pseudo posterior to achieve the same result in Equation~\ref{global} because the support is still the entire space, $\Theta$, for any finite $n$.

The local $\Delta_{\bm{\alpha}, \mathbf{x}}$ in Equation~\ref{local} is a random quantity based on only the observed values $x_i$ from $n$ draws from the generating distribution $P(x|\theta_0)$, and a fixed finite number $S$ of draws $\theta_s$ from the posterior  $\xi^{\alpha}(\theta|\mathbf{x})$. Finally, $\Delta^{\infty}_{\bm{\alpha},\mathbf{x}}$ in Equation~\ref{final} expresses the bound of the log-likelihood based on the observed database values evaluated at the limiting point $\theta^*$.

\begin{align}
 \Delta_{\bm{\alpha},\mathbf{x}} &= \max_{\theta_s \sim \xi^{\alpha}(\theta|\mathbf{x})} \left\{\inf \left\{w: \abs{f^{\bm{\alpha}(x_i)}_{\theta}(x_i)} \leq w, i \in \{1, \ldots n\}, x_i \sim P(x|\theta_0) \right\} \right\}\label{local}\\
 %\Delta^{\infty}_{\bm{\alpha},\mathbf{x}} & = \inf \left\{w: \abs{f^{\bm{\alpha}(x)}_{\theta^{\ast}}(x)} \leq w, i \in \{1, \ldots n\}, x_i \sim P(x|\theta_0) \right\} \label{final}  
    \Delta^{\infty}_{\bm{\alpha},\mathbf{x}} & = \inf \left\{w: \abs{f^{\bm{\alpha}(x_i)}_{\theta^{\ast}}(x_i)} \leq w, i \in \{1, \ldots \infty\}, x_i \sim P(x|\theta_0) \right\} \label{final}  
\end{align}

%over the truncated support for $\tilde{x}$.
%
Then based on Theorem~\ref{th:consistency} both
$P_{\theta_0}\{\abs{\Delta_{\bm{\alpha},\mathbf{x}} - \Delta^{\infty}_{\bm{\alpha},\mathbf{x}}} > 0 \}  \rightarrow 0$
and 
$P_{\theta_0}\{\abs{\Delta_{\bm{\alpha}} - \Delta^{\infty}_{\bm{\alpha},\mathbf{x}}} > 0 \}  \rightarrow 0$
because the pseudo posterior degenerates to a point mass at $\theta^*$. 
Thus
\begin{equation}
P_{\theta_0}\{\abs{\Delta_{\bm{\alpha},\mathbf{x}} - \Delta_{\bm{\alpha}}} > 0 \}  \rightarrow 0
\end{equation}

Assumption~\ref{ass:riskweight} ensures a formal privacy guarantee since $\Delta_{\bm{\alpha}} < \infty$, by construction, and the asymptotic result guarantees that the local $\Delta_{\bm{\alpha},\mathbf{x}}$ will get arbitrarily close to the global $\Delta_{\bm{\alpha}}$ where $\epsilon = 2\Delta_{\bm{\alpha}}$.  For a large $n$, then, $\Delta_{\bm{\alpha},\mathbf{x}} \rightarrow \Delta_{\bm{\alpha}}$ becomes independent of $\mathbf{x} \in \mathcal{X}^{n}$, where we recall that $\epsilon = 2\Delta_{\alpha}$.  This contraction of the local Lipschitz bound onto a global value that determines the privacy guarantee, $\epsilon$, indicates that $\pi$ of our $(\epsilon,\pi)-$ pDP guarantee in Equation~\ref{eq:pDP} contracts onto $0$ at $\mathcal{O}(n^{-1/2})$ rate at which $\Delta_{\bm{\alpha},\mathbf{x}}$ contracts onto $\Delta_{\bm{\alpha}}$. 
To speed convergence and add stability for $\Delta_{\bm{\alpha}, \mathbf{x}}$ for finite $n$, we consider employing a threshold $M$ for the $\alpha$-weighted log-likelihood, such that $\abs{f^{\bm{\alpha}(x_i)}_{\theta}(x_i)} > M$ is replaced by 0, through setting $\bm{\alpha}(x_i) = 0$. 
%This allows us to indirectly truncate the support of $\theta$ for finite sample sizes $n$.

To make intuitive the rate of contraction of $\pi$ to $0$ at $\mathcal{O}(n^{-1/2})$, we conduct a Monte Carlo simulation study next to develop a distribution of local Lipschitz bounds from which we compute the local Lipschitz, $\Delta_{\bm{\alpha},\mathbf{x}}$, each at an increasing sequence of sample sizes, $n$.  We reveal that the distribution over local Lipschitz bounds contracts together onto a single global value, demonstrating the local-to-global contraction as $n$ increases. In particular, the use of the $M$ threshold greatly stabilizes and speeds convergence.

\subsection{Asymptotic Differential Privacy Guarantee}
We have noted that Theorem~\ref{th:consistency} induces the contraction of $\Delta_{\bm{\alpha},\mathbf{x}}$ computed on database, $\mathbf{x}$, to the global Lipschitz bound, $\Delta_{\bm{\alpha}}$.  This contraction is driven by the collapsing of the parameter space, $\Theta$, to a point, $\theta^{\ast}$, asymptotically in data size, $n$.   

Our implementation for computing the by-record Lipschitz bound, $\displaystyle\mathop{\sup_{\theta\in\Theta}}\vert f(x_{i}\mid\theta)\vert$, on a database relies on this asymptotic convergence.  For implementation on a database we evaluate
$\displaystyle\mathop{\sup_{\theta \in \xi^{\alpha}(\theta|\mathbf{x})}}\vert f(x_{i}\mid\theta)\vert$; that is, we compute the supremum of the absolute value of the log-likelihood over the subset of $\Theta$ that receives positive posterior mass.   This subset shrinks to a point for $n$ sufficiently large, making our treatment conservative for large $n$. Our de facto truncation of $\Theta$ to that subset receiving positive posterior measure is similar to the probabilistic Lipschitz condition of Assumption 2 in \citet{Dimitrakakis:2017:DPB:3122009.3122020}; only, in our implementation of differential privacy guarantee is asymptotic such that we rely on the shrinking of the size of $\Theta$ with increasing $n$ to state our privacy guarantee. Our use of a threshold $M$ speeds this convergence, making it useful for moderate sample sizes. This indirect truncation of $\Theta$ is much simpler to implement compared to specifying a meaningful truncated prior distribution in high-dimension.

\subsection{Simulation Study}\label{asympt:sim}

We next utilize a Monte Carlo simulation study by fixing a sample size, $n$,  and repeatedly generating a count data sample from a Poisson generating model. We proceed to compute the local Lipschitz bound for each sample database for the $\bm{\alpha}-$weighted pseudo posterior mechanism and also the unweighted posterior synthesizer to provide a comparison.  This procedure gives us a distribution of the local Lipchitz bounds across databases of size $n$. We repeat this process for an increasing sequence of sample sizes, $\mathbf{n} = 100*4^{(0,1,2,3)} = (100,400,1600,6400)$.

In addition to computing the local Lipschitz bounds at each $n$ for the $\bm{\alpha}-$weighted pseudo posterior mechanism, we introduce an extension to our pseudo posterior mechanism that truncates the weight, $\alpha_{i}$, for each likelihood contribution in the following procedure:
\begin{enumerate}
\item Compute weights, $\bm{\alpha}$, for local database, $\mathbf{x}$, using the procedure of Section~\ref{sec:setweights}.  We first compute $f_{i}$ (the maximum of absolute log-likelihood values for record, $i$, over the sampled values of $\theta_{s}$) for each database record, $i \in (1,\ldots,n)$, from the unweighted posterior mechanism. Then using the linear transform, $\alpha_{i} = 1 - \tilde{f}_{i}$, where $\tilde{f}_{i} = \frac{f_i - \min_j f_j}{\max_j f_j - \min_j f_j}$.
\item We add a step to truncate the weight for any record whose weighted log-pseudo likelihood value is greater than some threshold, $M$, to $0$, completely removing the likelihood contribution for record $i$.  We accomplish this truncation by forming a weighted absolute log-pseudo likelihood for each record, $i$,  as $\alpha_{i} \times f_{i}$. If $\alpha_{i} \times f_{i} > M$, we set final weight, $\alpha^{\ast}_{i} = 0$; otherwise we leave $\alpha^{\ast}_{i} = \alpha_{i}$ unchanged.  The motivation for this method is to more tightly control or correct the local Lipschitz to $M$ that we will observe in the sequel speeds convergence.
We choose $M$ based on oracle information based on experience with databases of similar types.

\end{enumerate}
The use of a threshold, $M$, to truncate weights is a stricter implementation from our weight-setting procedure of Assumption~\ref{ass:riskweight}.  We recall that this assumption guarantees the existence of a global Lipschitz because for every database it sets the weight for a record with a non-finite absolute log likelihood to $0$.  In this stricter implementation, we set $\alpha^{\ast}_{i} = 0$ if its weighted absolute log-pseudo likelihood is $> M$, where we choose $M$ based on oracle information based on experience with databases of similar types.

Using the means model for Poisson distributed data, $y \sim Pois(\mu)$ (with $\mu = 100$) our simulation procedure is, as follows.
%replace with algorithm block? or tigtenup the spacing
\begin{enumerate}%[noitemsep]
\item For sample size, $n \in \{100,400,1600,6400\}$, repeat the following Monte Carlo procedure to generate a distribution of local Lipschitz bounds:
\item For $r = 1,\ldots, 400$:
\begin{itemize}[noitemsep]
	\item Generate $\mathbf{y}_r \sim \text{Pois}(\mu)$, each of size $n$.
	\item Compute the \emph{local} Lipschitz bound, $\Delta_{\bm{\alpha},\mathbf{y}}$, for the unweighted, $\bm{\alpha}-$weighted, and $M-$truncation-weighted pseudo posterior mechanisms.
    \item  Construct the distribution of $\Delta_{\bm{\alpha},\mathbf{y}_{r}}$ and note the maximum of the distribution and difference between the maximum and minimum values of the distribution of the local Lipschitz bounds at each sample size, $n$.
\end{itemize}
\item Assess contraction of the $\max_{r}\Delta_{\bm{\alpha},\mathbf{y}_{r}}$ to a single (global) value and whether the minimum and maximum values collapse together.
\end{enumerate}

To assess the \emph{contraction} of the maximum point in the distribution of \emph{local} Lipschitz bounds to the \emph{global} Lipschitz bound, we repeat the simulation above using sample sizes $\mathbf{n} = (100, 400, 1600,6400)$. Figure \ref{fig:lockLmain} compares the distributions across the $R = 400$ replications for the unweighted (labeled as ``Unweighted"), $\bm{\alpha}-$weighted (labeled as ``Weighted") that does \emph{not} use truncation of weights and the truncated weighted at $M$ (labeled as ``Weighted-M") mechanisms, from left-to-right.  The distribution of local Lipschitz bounds for the Unweighted mechanism increases (or drifts) with larger sample sizes. The Weighted mechanism (that includes no weight truncation) shows a pronounced decrease in drift in the maximum Lipschitz of local databases over the increasing sample sizes as compared to the Unweighted mechanism, though even at sample size, $6400$, there is still a small, though decreasing drift of the maximum Lipschitz. By contrast, the Weighted-M mechanism, under setting $M = 3.5$, demonstrates rapid contraction of both the minimum and maximum local Lipschitz values onto $M$. This is still a probabilistic formal privacy result because the local Lipschitz values are not strictly bounded below $M$ due to sampling variability of $\theta$. The maximum of the distribution of local Lipschitz bounds at each sample size is slightly larger than $M$, indicating that our guarantee is probabilistic. While both the Weighted and Weighted-M local Lipschitz bounds contract at $\mathcal{O}(n^{-1/2})$, the multiplicative constant of the contraction rate is much smaller for Weighted-M because of the truncation to an asymptotic global Lipschitz of $M$ defined by the owner of the closely-held data.

Figure \ref{fig:utilities} presents the distributions for the averages of the mean parameter, $\mu$, over the $R=400$ Monte Carlo iterations.  We see there is some utility loss relative to Unweighted and Weighted under use of Weighted-M, though the resulting utility is still relatively robust.  The deterioration of the utility for Weighted-M as $n$ increases (as represented by the growing dissimilarity of the pseudo posterior distribution for $\mu$ to that under Unweighted) is a conservative result because we use the same $M = 3.5$ for all sample sizes.  Yet, the DP guarantee is based on the space of databases at a particular sample size, $n$, and $M$ will be set based on agency experience with a particular class of data (e.g., monthly survey responses) that all have very similar values for $n$.

\begin{figure}[htbp]
\centering
\includegraphics[width=0.9\textwidth, page = 1]{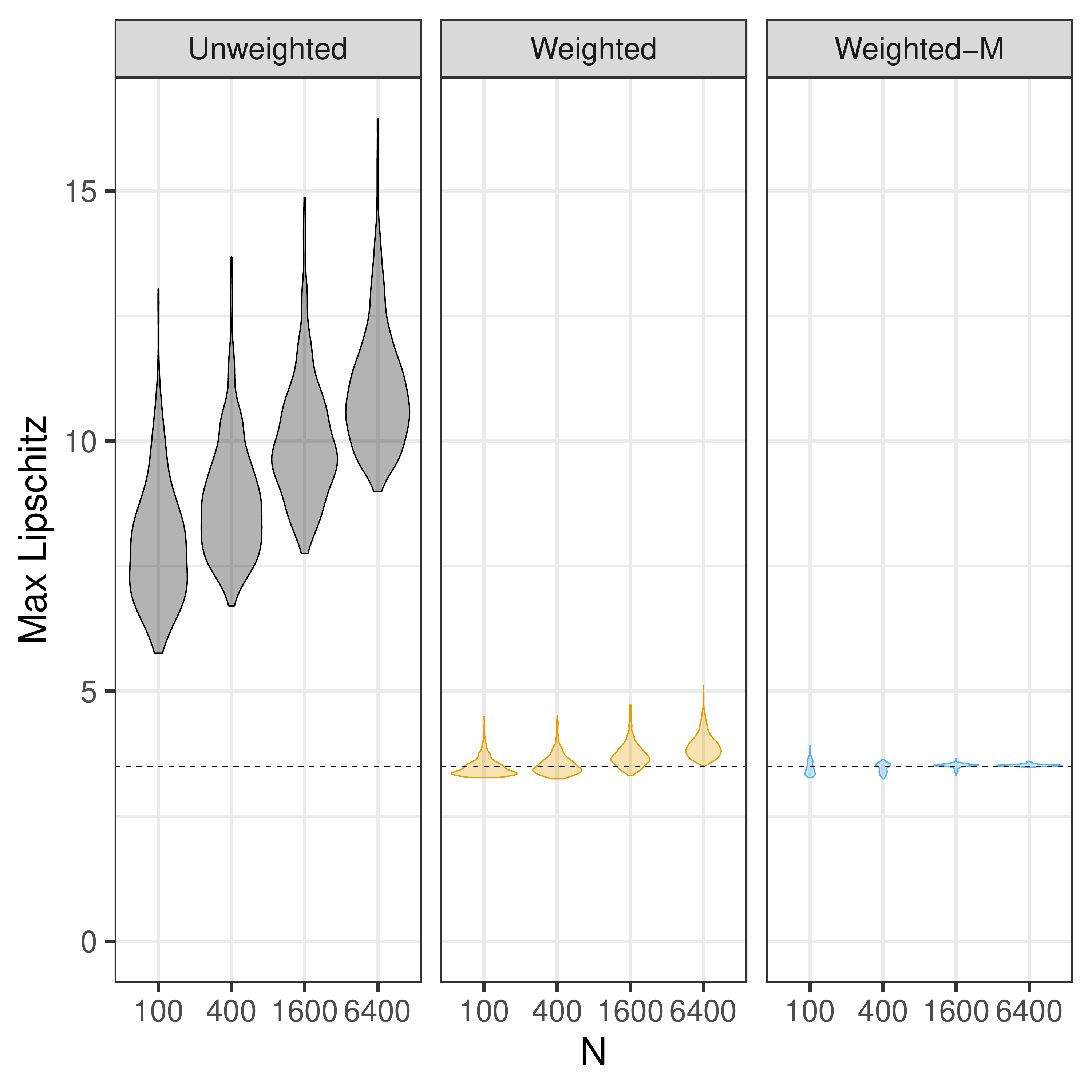}
\caption{Distribution of the maximum observed Lipschitz bound $\Delta_{\mathbf{y}}$ for each of sample sizes $(100, 400, 1600,6400)$  from $R=400$ realizations of pseudo posterior samples of (left to right) unweighted, weighted, and weighted-M ($M-$ truncated weights).}
\label{fig:lockLmain}
\end{figure}
\FloatBarrier

\begin{figure}[htbp]
\centering
\includegraphics[width=0.9\textwidth, page = 1]{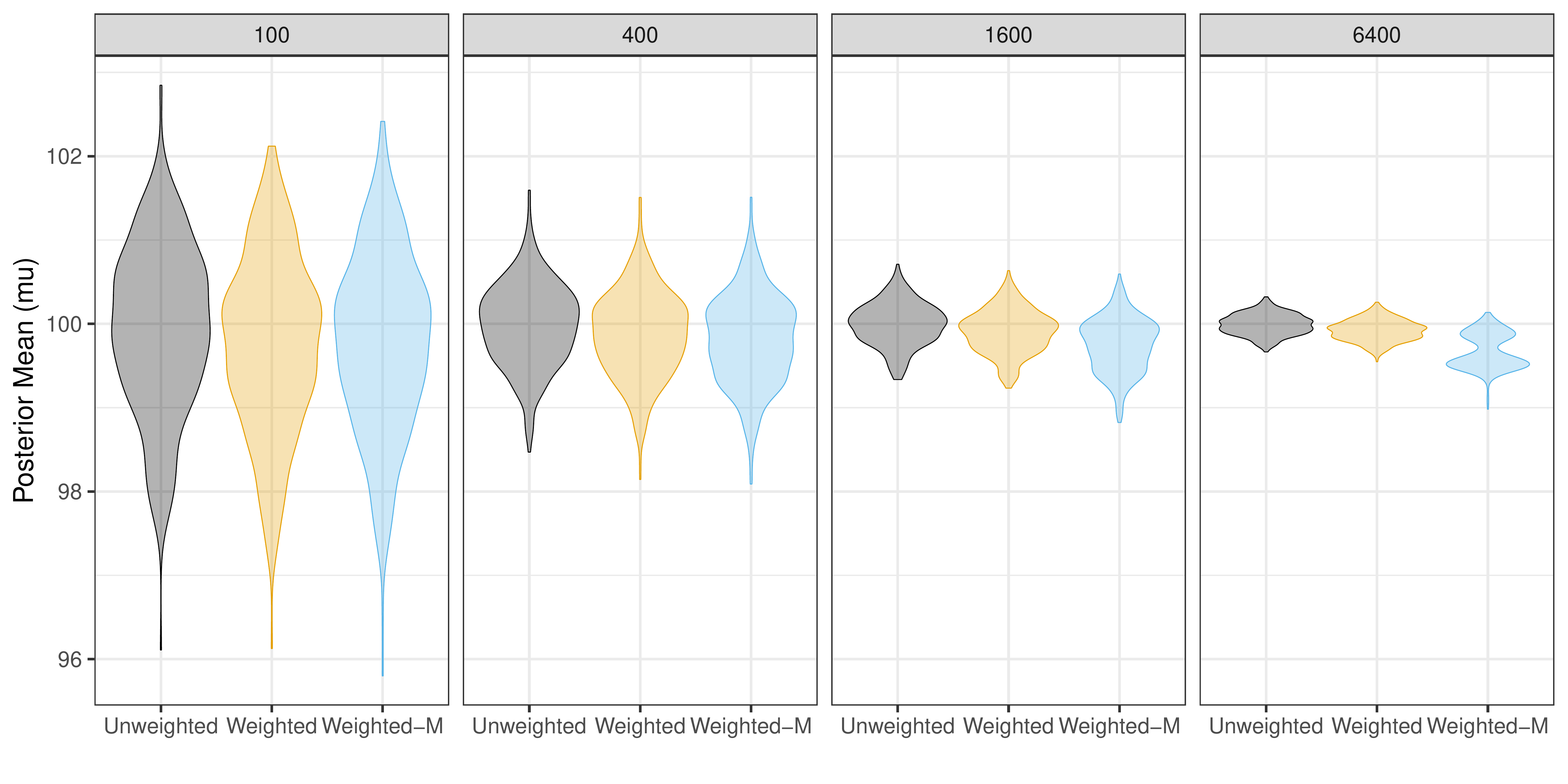}
\caption{Distributions of the average of mean parameter $\mu$ for each of sample size $(100, 400, 1600, 6400)$  from $R=400$ realizations of pseudo posterior samples of (left to right) unweighted, weighted, and weighted-M ($M-$ truncated weights). }
\label{fig:utilities}
\end{figure}
\FloatBarrier

\subsection{Recommendation for Setting a Global $\epsilon$ from a weighted$-M$ Lipschitz}
Although we have demonstrated an $\mathcal{O}(n^{-1/2})$ contraction rate of $\pi$ (the probability of exceeding $\epsilon-$ global DP), in theory, and have further illustrated this convergence in our above simulation study, it is difficult in practice to discover at what sample size under a specific synthesizer that one may declare the local Lipschitz to be global.  For typically used sample sizes $> 1000$ we suggest to take the Weighted-M Lipschitz and employ a multiplicative ``factor of safety", $s \in (1,1.05)$, to develop an upper, global bound that, in turn, determines $\epsilon$ because the contraction is extremely rapid.  If the sample size is $< 1000$, we recommend to set multiplicative $s' \in (1.05,1.10)$.

\section{Application to the CE Sample}
\label{sec:app}

We introduce the CE sample of consumer units (CU) or households in Section \ref{sec:app:CEdata}, where our goal is to synthesize a highly-skewed continuous variable, family income, under a local DP guarantee provided by our $\bm{\alpha}-$weighted pseudo posterior mechanism. In Section \ref{sec:app:utilityrisk}, we present risk and utility profiles of synthetic data drawn from our $\bm{\alpha}-$weighted pseudo posterior mechanism, along with comparisons to the EM, the non-differentially private risk-weighted synthesizer of \citet{HuSavitskyWilliams2021rebds} and the unweighted posterior mechanism.  Section \ref{sec:app:cg} presents privacy and utility results with different scaling and shifting, $(c, g)$, configurations for vector weights in Equation (\ref{eq:vectroweights}) to sketch out a risk-utility curve for our $\bm{\alpha}-$weighted pseudo posterior mechanism that we compare to that of the EM. A risk-utility curve provides the Bureau of Labor Statistics (BLS) options for selecting a risk-utility setting that matches their policy objectives.

\subsection{The CE Sample and Unweighted Synthesizer}
\label{sec:app:CEdata}

Our application of the $\bm{\alpha}-$weighted pseudo posterior mechanism focuses on providing privacy protection for a family income variable published by the CE.  The CE is administered by the BLS with the purpose of providing income and expenditure patterns indexed by geographic domains to support policy-making by State and Federal governments.  The description of the CE sample included here closely follows that in \citet{HuSavitskyWilliams2021rebds}. The CE contain data on expenditures, income, and tax statistics about CUs across the U.S. The CE public-use microdata (PUMD)\footnote{For for information about CE PUMD, visit {\url{https://www.bls.gov/cex/pumd.htm}}.} is publicly available record-level data, published by the CE. The CE PUMD has undergone masking procedures to provide privacy protection of survey respondents. Notably, the family income variable has undergone top-coding, a popular Statistical Disclosure Limitation (SDL) procedure that may result in reduced utility and insufficient privacy protection \citep{AnLittle2007JRSSA, HuSavitskyWilliams2021rebds}.

The CE sample in our application contains $n = 6208$ CUs, coming from the 2017 1st quarter CE Interview Survey. It includes the family income variable, which is highly right-skewed and deemed sensitive; see Figure \ref{fig:FamilyIncome} for its density plot. The CE sample also contains 10 categorical variables, listed in Table \ref{tab:CEvars}. These categorical variables are deemed insensitive and used as predictors in building a flexible synthesizer for the synthesis of the sensitive family income variable.

\begin{figure}[htbp]
  \centering
    \includegraphics[width=0.5\textwidth]{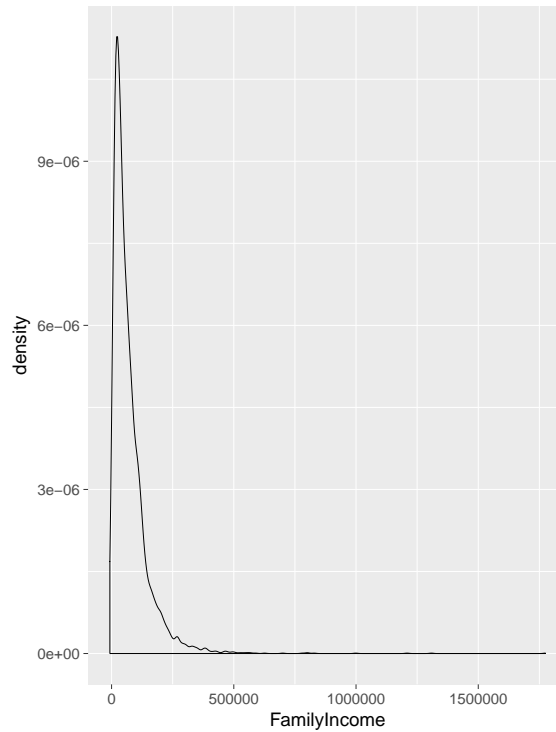}
    \caption{Density plot of Family Income in the CE sample. \label{fig:FamilyIncome}}
\end{figure}

\begin{table}[H]
\centering
%{\footnotesize{
\caption{Variables used in the CE sample. Data taken
  from the 2017 Q1 Consumer Expenditure Surveys. \label{tab:CEvars}}
\begin{tabular}{ll}
\hline
Variable &  Description\\ \hline
Gender & Gender of the reference person; 2 categories \\
Age & Age of the reference person; 5 categories  \\
Education Level & Education level of the reference person; 8 categories  \\
Region & Region of the CU; 4 categories \\
Urban & Urban status of the CU; 2 categories  \\
Marital Status & Marital status of the reference person; 5 categories  \\
Urban Type & Urban area type of the CU; 3 categories  \\
CBSA & 2010 core-based statistical area (CBSA) status; 3 categories \\
Family Size & Size of the CU; 11 categories \\
Earner & Earner status of the reference person; 2 categories  \\
Family Income & Imputed and reported income before tax of the CU; \\
&  approximate range: (-7K, 1,800K) \\ \hline
\end{tabular}
%}}
\end{table}

To generate partially synthetic databases for the CE sample with synthetic family income, we use an unweighted, non-private synthesizer: a flexible, parametric finite mixture synthesizer.
%The synthesizer is constructed as an over-determined finite mixture using a prior distribution for probabilities of assignment to each cluster that encourages sparsity in the number of populated clusters such that our parametric model becomes arbitrarily close to the Dirichlet Process mixture in the limit of the maximum number of clusters, $K$.
This finite mixture synthesizer has been shown to produce synthetic data characterized by a high utility, but also with an unacceptable level of disclosure risk \citep{HuSavitskyWilliams2021rebds}. We leave the details of the synthesizer in the Appendix \ref{appendix:synthesizer} for brevity and direct interested readers to \citet{HuSavitskyWilliams2021rebds} for further information.

\subsection{Risk and Utility Comparisons}
\label{sec:app:utilityrisk}

To generate synthetic data and compare results, we apply four synthesizers: 1) the unweighted, non-(locally) private synthesizer, labeled ``Unweighted"; 2) the locally private synthesizer under the $\bm{\alpha}-$weighted pseudo posterior mechanism, labeled ``DPweighted", with configuration $(c, g) = (0.7, 0.0)$; 3) the locally private synthesizer under the EM, labeled ``EMweighted", which is designed to privacy target, $\epsilon$, achieved by ``DPweighted"; 4) and the weighted, though non-(locally) private pseudo posterior synthesizer proposed by \citet{HuSavitskyWilliams2021rebds}, labeled ``Countweighted", that utilizes their method for measuring the by-record disclosure risk (based on an assumption about the behavior of an intruder). We use $\bm{\alpha}_c$ to denote the risk-adjusted record-indexed weights calculated in the Countweighted method. The labels are used throughout the remainder of this paper when presenting various risk and utility results.
%\begin{figure}[H]
%\centering
%\includegraphics[width=0.9\textwidth]{Lipschitz_real_linear_990_0p6_0_lg}
%\caption{Lipschitz, $c = 0.6, \gamma = 0$. $\Delta_{Unweighted} = 78.7, \Delta_{DPweighted} = 8.16, \Delta_{EMweighted} = 10.4$.}
%\label{fig:L0p6}
%\end{figure}

We first look at the risk profiles of the four synthesizers. Figure \ref{fig:L0p7} plots the distributions of the Lipschitz bounds, $\Delta_{x_{i}}$'s, for each of the four synthesizers computed by taking the maximum of the $S$ log-likelihood ratios for each record, $i = 1,\ldots,(n = 6208)$ over the $S$ draws of $\theta$ from it's posterior distribution.  The maximum value of the $(\Delta_{x_i})$ over all of the records is denoted as $\Delta_\mathbf{x}$, the Lipschitz bound for the mechanism.

\begin{figure}[htbp]
\centering
\includegraphics[width=0.8\textwidth]{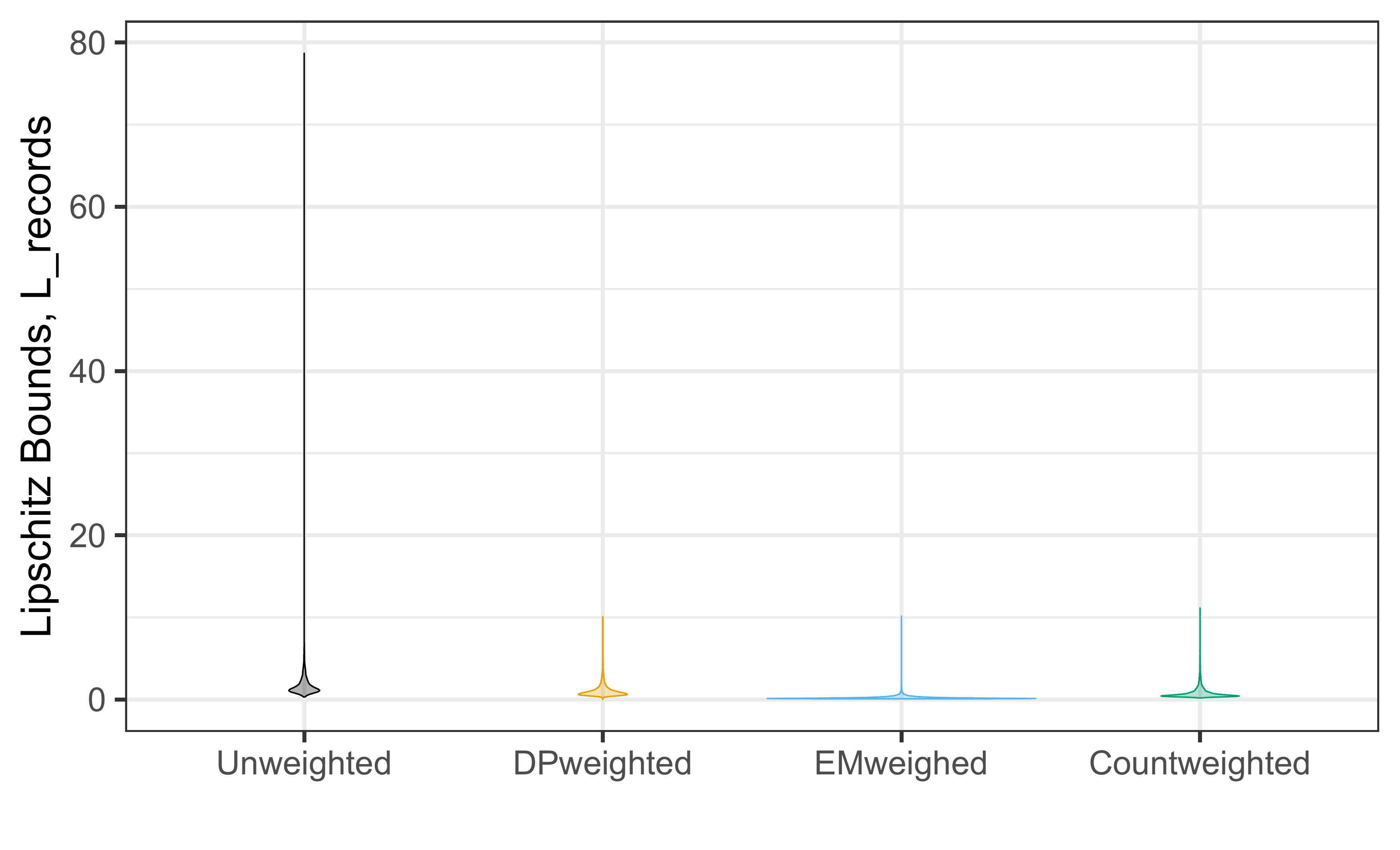}
\caption{Violin plots of the distribution of the Lipschitz bounds, $\Delta_\mathbf{x}$'s, for synthetic data generated the four synthesizers. The corresponding maximum $\Delta_\mathbf{x}$ values are: $\Delta_{Unweighted} = 78.7, \Delta_{\bm{\alpha},DPweighted} = 10.1, \Delta_{EMweighted} = 10.2, \Delta_{\bm{\alpha}_{c},Countweighted} = 11.17$.}
\label{fig:L0p7}
\end{figure}

The Unweighted, non-private synthesizer clearly has the highest maximum $\Delta_\mathbf{x}$ with $\Delta_{Unweighted} = 78.7$. The other non-private Countweighted synthesizer achieves a much lower maximum $\Delta_\mathbf{x}$ with $\Delta_{\bm{\alpha}_{c},Countweighted} = 11.17$. The large reduction in the Countweighted synthesizer owes to the positive correlation between by-record weights, $\bm{\alpha}_{c}$, where each $\alpha_{ci}$ is computed as the probability that the value for each target record is relatively isolated from that of other records used in the Countweighted synthesizer, on the one hand, with the by-record log-pseudo likelihood ratio bounds used for the DPweighted mechanism, on the other hand. DPweighted denotes our $\bm{\alpha}-$pseudo posterior mechanism.  The two locally private synthesizers both achieve even lower maximum $\Delta_\mathbf{x}$: $\Delta_{\bm{\alpha},DPweighted} = 10.1, \Delta_{\bm{\alpha},EMweighted} = 10.2$, indicating the best risk profiles. The EMweighted mechanism was estimated by setting the scalar with a target $\epsilon_{\mathbf{x}} = 2\Delta_{\bm{\alpha},\mathbf{x}}$, the local privacy guarantee (expenditure) achieved by our DPweighted mechanism with Lipschitz $\Delta_{\bm{\alpha},\mathbf{x}}$.  Our intent is to compare the utility performances between the two private mechanisms (DPweighted and EMweighted) where each achieves an equivalent privacy guarantee.  It bears mention that while the DPweighted under the pseudo posterior mechanism and the EMweighted under the EM achieve similar maximum local Lipschitz bounds, which governs the local DP guarantee, the EM tends to produce notably lower risk for most records than the DPweighted mechanism.  The lower record-indexed risk for EMweighted as compared to DPweighted is evident in the flattened shape of the violin plot for EMweighted. The EM sets the scalar weight based on the risk of the worst case over all records because the same level of downweighting must be applied to all records in contrast with the by-record weighting under of our $\bm{\alpha}-$weighted pseudo posterior mechanism in DPweighted.

%\begin{figure}[H]
%\centering
%\includegraphics[width=0.9\textwidth]{mean_real_linear_990_0p6_0_lg}
%\caption{mean, $c = 0.6, \gamma = 0$}
%\label{fig:mean0p6}
%\end{figure}

Figure \ref{fig:mean0p7} and Figure \ref{fig:90q0p7} show a collection of violin plots of the distribution (obtained from re-sampling) for each of the mean and the 90th quantile statistics, respectively, estimated on the synthetic data generated under each of our four synthesizers and also on the closely-held confidential (real) data for comparison, labeled ``Data".  These figures allow us to compare the utility performances across our synthesizers by the examination of how well the real data distribution for each statistic is reproduced by the synthetic database for each of our synthesizers.  For the synthesizers, a set of $m = 20$ synthetic databases were generated and the distribution for each statistic was estimated on each databases (under re-sampling).  The resulting barycenter of the individual distributions in the Wasserstein space of measures was computed by averaging the quantiles over the $M$ databases \citep{Srietal15}.  Our privacy guarantees apply to \emph{each} synthetic draw from our mechanism, so the total privacy expenditure is that for each database shown in Figure~\ref{fig:L0p7} multiplied by $m$.   We compute utilities over $m = 20$ synthetic databases to fully capture the uncertainty in the synthetic data generation process from the (pseudo) posterior predictive distributions. Generating multiple synthetic databases are also standard practice in the research and practice of synthetic data using Bayesian synthesizers \citep{ReiterRaghu2007}. We note that the distribution of each statistic for a single synthetic database is very similar.

\begin{figure}[htbp]
\centering
\includegraphics[width=0.8\textwidth]{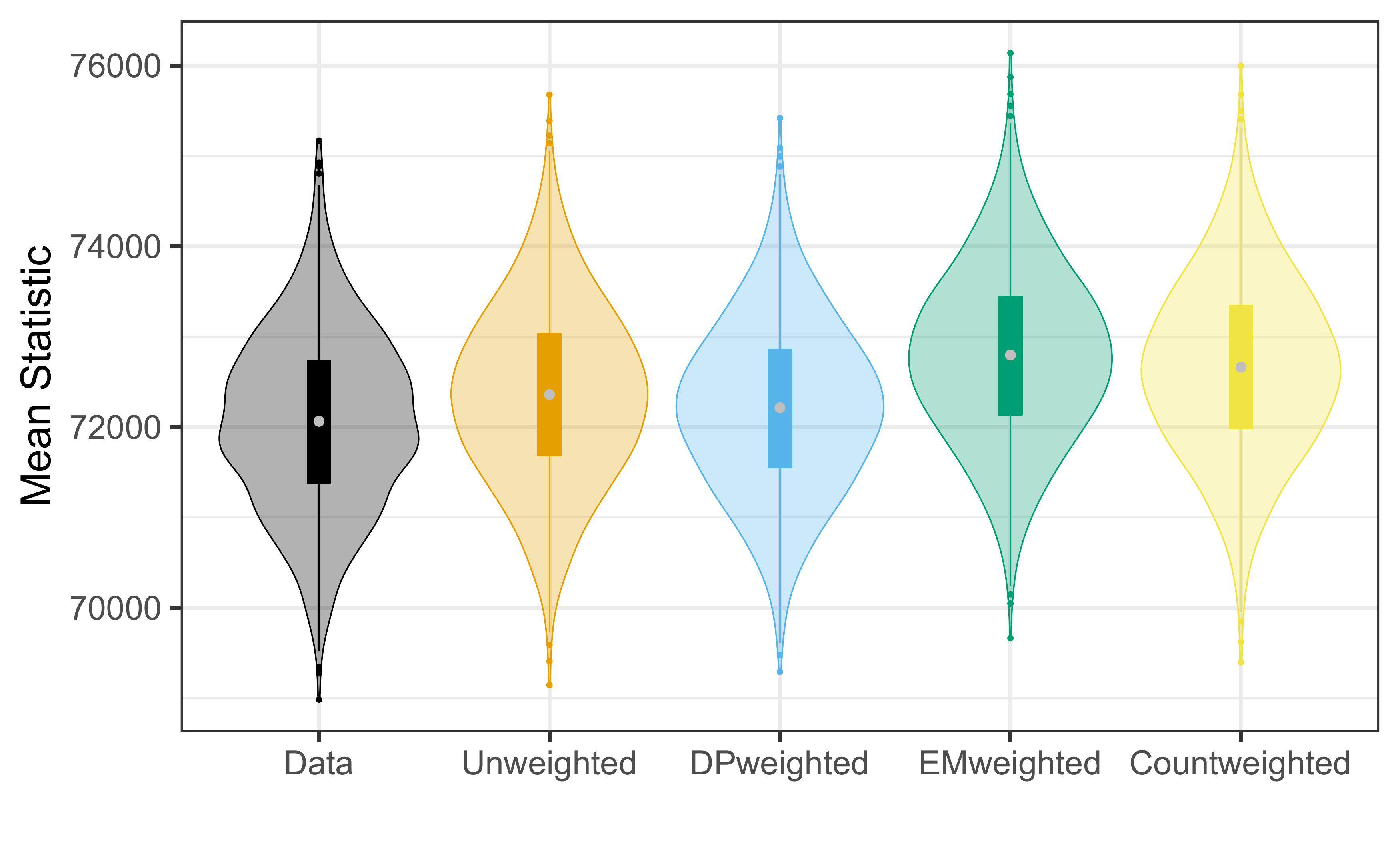}
\caption{Violin plots of the mean estimation of the confidential CE sample and the four synthesizers.}
\label{fig:mean0p7}
\end{figure}

%\begin{figure}[H]
%\centering
%\includegraphics[width=0.9\textwidth]{median_real_linear_990_0p6_0_lg}
%\caption{median, $c = 0.6, \gamma = 0$}
%\label{fig:median0p6}
%\end{figure}

\begin{comment}
\begin{figure}[H]
\centering
\includegraphics[width=0.9\textwidth]{median_real_linear_990_0p7_0_lg}
\caption{median, $c = 0.7, g = 0$}
\label{fig:median0p7}
\end{figure}
\end{comment}

%\begin{figure}[H]
%\centering
%\includegraphics[width=0.9\textwidth]{q90_real_linear_990_0p6_0_lg}
%\caption{90 quantile, $c = 0.6, \gamma = 0$}
%\label{fig:90q0p6}
%\end{figure}

\begin{figure}[htbp]
\centering
\includegraphics[width=0.8\textwidth]{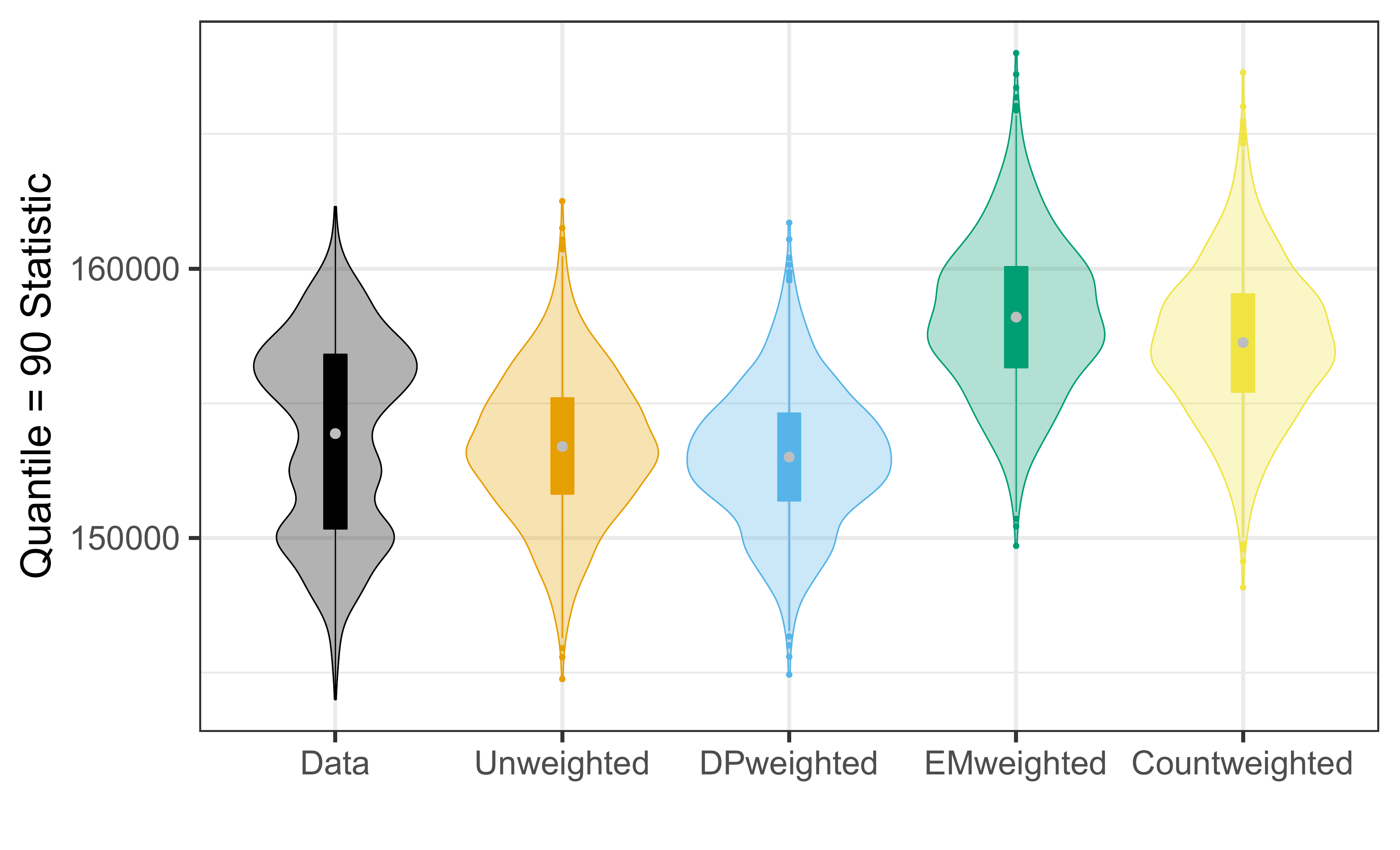}
\caption{Violin plots of the 90th quantile estimation of the confidential CE sample and the four synthesizers.}
\label{fig:90q0p7}
\end{figure}

The DPweighted synthesizer outperforms the EMweighted and Countweighted mechanisms in utility preservation. First, especially evident in Figure \ref{fig:90q0p7}, DPweighted (the $\bm{\alpha}-$weighted pseudo posterior mechanism) provides better estimates than EMweighted (the scalar-weighted EM). The notably deteriorated utility preservation of the EM derives from the setting that scalar weight applied to all records based on the highest risk records as earlier discussed. Since both mechanisms achieve the same maximum Lipschitz bound $\Delta_\mathbf{x}$, which governs the local DP guarantee, these results indicate that the EM has to compromise a large amount of the utility to achieve a similar local DP guarantee compared to the $\bm{\alpha}-$weighted pseudo posterior mechanism.

Second, while the non-private Unweighted synthesizer and the locally private DPweighted synthesizer provide equally good estimates for both the mean and the 90th quantile, the much greater Lipschitz bound of the Unweighted synthesizer shown in Figure \ref{fig:L0p7} indicates a much worse balance for the utility-risk trade-off as compared to DPweighted. The third minor point is that the Countweighted synthesizer, albeit non-locally private, achieves only a slightly higher maximum Lipschitz bound compared to our private DPweighted synthesizer. However, its utility preservation is worse, especially evident in Figure \ref{fig:90q0p7} for the 90th quantile estimation.

In summary, our private DPweighted mechanism outperforms the other three synthesizers to achieve a highly satisfactory risk-utility trade-off balance. We next explore different scaling and shift configurations of $(c, g)$, introduced in Section \ref{sec:setweights}, to sketch out the risk-utility curves for DPweighted and EMweighted.

\subsection{Mapping DP Risk and Utility Curves}
\label{sec:app:cg}

We conclude by applying the scaling parameter, $c$, and the shift parameter, $g$, to the distribution of weights, $\bm{\alpha}$, used in our $\bm{\alpha}-$weighted pseudo posterior mechanism in order to enumerate the risk-utility curve to support the choice of $\bm{\alpha}$ (and, hence, $\Delta_{\bm{\alpha},\mathbf{x}}$, and $\epsilon_{\mathbf{x}}$). Having such a risk-utility curve would allow the BLS (or, more generally, the owner of the closely-held private database) to discover the setting configuration that best represents their policy goal for the level of privacy protection sought.  We compare the risk-utility mapping produced by the $\bm{\alpha}-$weighted pseudo posterior mechanism to that of the EM, which we recall reduces to a scalar-weighted pseudo posterior under use of the log-likelihood as the utility measure.  As discussed in \citet{HuSavitskyWilliams2021rebds}, applying a scaling constant, $c < 1$, will induce a compression in the distribution of the weights while apply a scaling $g < 0 $ will induce a downward shift in the distribution of the record-indexed weights.   We apply the scaling and shifting in a manner that uses truncation to ensure each of the resulting weights are restricted to lie in $[0, 1]$.

Each violin plot in Figure~\ref{fig:q90compare} presents a distribution of the 90th quantile for a synthetic database generated under a particular $(c, g)$ configuration.  The sequence of plots from left-to-right are ordered from less scaling and shifting (with a relatively higher or looser level for the privacy guarantee) to more scaling and shifting (with a relatively lower  or tighter level for the privacy guarantee).   The specific local sensitivity or Lipschitz value, $\Delta_{\bm{\alpha},\mathbf{x}}$, associated with each configuration are shown in Table~\ref{tab:cg}, where we recall that the associated local privacy guarantee is $\epsilon_{\mathbf{x}} = 2 \times \Delta_{\bm{\alpha},\mathbf{x}} \times (m=20)$, where the multiplication by $m=20$ derives from our use of multiple posterior draws to generate multiple synthetic databases.

The accompanying Table~\ref{tab:cg} demonstrates a nearly $80\%$ reduction in the level for the local DP guarantee of the $\bm{\alpha}-$weighted pseudo posterior mechanism over the range of configurations.

Figure~\ref{fig:q90compare} plots the distribution of the $90-$th quantile for the generated synthetic data under each of the Unweighted (``UW"), Exponential (``EM") and $\bm{\alpha}-$weighted pseudo posterior (``DP") mechanisms at a sequence of (scaling, shifting), $(c,g)$, combinations.  The local sensitivity/Lipschitz, $\Delta_{\bm{\alpha},\mathbf{x}}$, is lower as one traverses left-to-right, indicating a stronger local privacy guarantee on the right-hand side.  This sequence of plots demonstrates a much flatter or reduced deterioration of the $90-$th quantile distribution for the DPweighted mechanism, the $\bm{\alpha}-$weighted pseudo posterior mechanism, as compared to the EMweighted mechanism.  The superior result for DPweighted is not surprising due to the greater flexibility of DPweighted to concentrate downweighting to high-risk records versus the application of a scalar weight based on the highest risk record to all records under EMweighted.

\begin{figure}[htbp]
\centering
\includegraphics[width=0.9\textwidth]{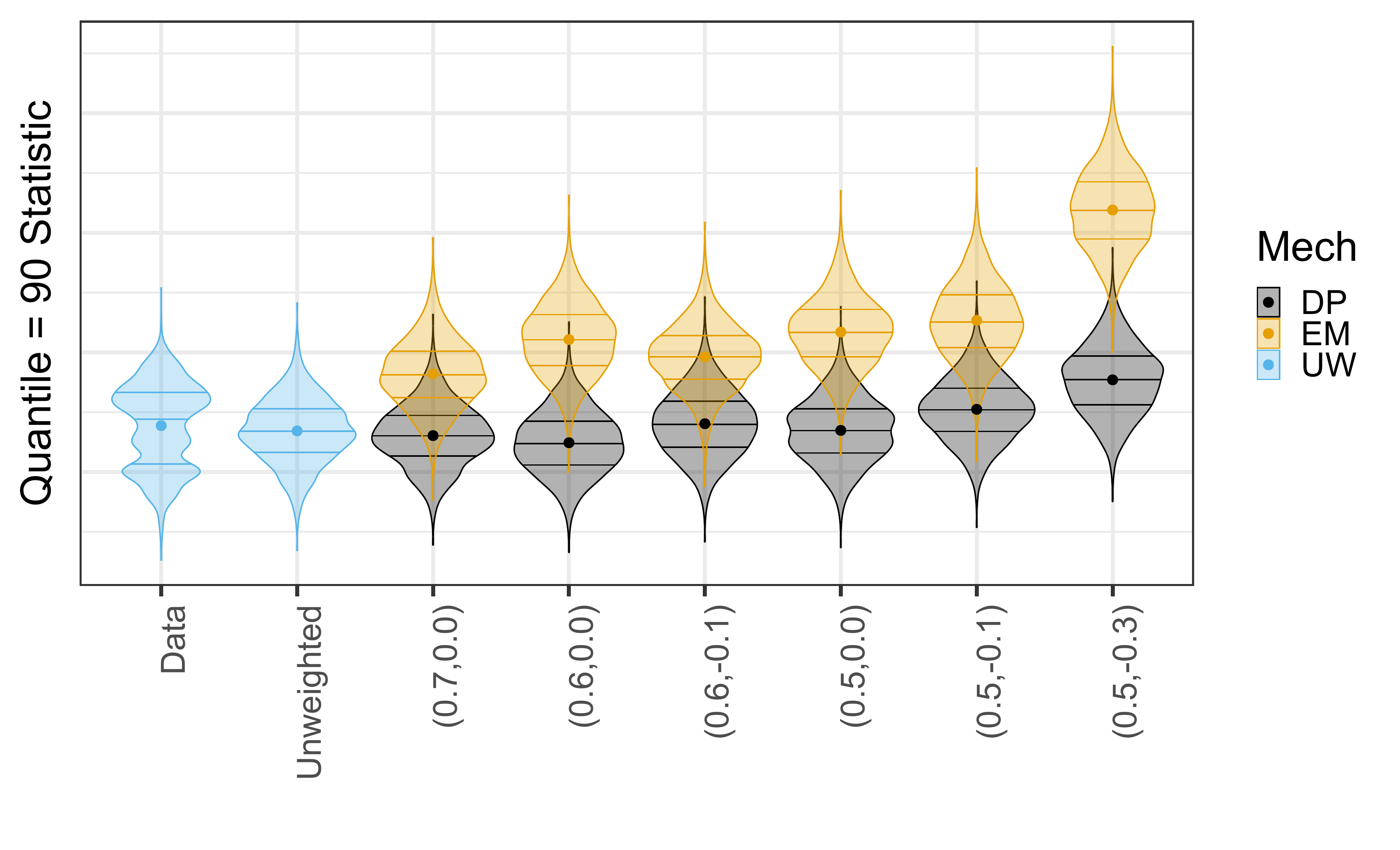}
\caption{Violin plots of the 90th quantile estimation of: 1) the confidential CE sample; 2) the unweighted, non-private synthesizer; and overlapping violin plots of the 90th quantile estimation of the synthesizer under the $\bm{\alpha}-$weighted pseudo posterior mechanism compared to the synthesizer under the EM with equivalent $\Delta_{\bm{\alpha}, \mathbf{x}}$ values, for the following $(c, g)$ configurations: 3) $(c, g) = (0.7, 0.0)$; 4) $(c, g) = (0.6, 0.0)$; 5) $(c, g) = (0.6, - 0.1)$; 6) $(c, g) = (0.5, 0.0)$; 7) $(c, g) = (0.5, - 0.1)$; 8) $(c, g) = (0.5, - 0.3)$.}
\label{fig:q90compare}
\end{figure}

%\begin{table}[H]
%\begin{center}
%\caption{Table of values of the Lipchitz bound $\Delta$, of the synthesizer under the pseudo posterior mechanism, for a series of $(c, g)$ configurations. $\Delta_{Synthetic} = 78.7.$}
%\label{tab:cg}
%\begin{tabular}{ c | c c c c c }
%\hline
%$(c, g)$  & (0.7, 0.0) & (0.6, 0.0) & (0.6, -0.1) & (0.5, 0.0) & (0.5, -0.1)\\
%$\Delta$ & 10.10 & 8.16 & 7.30 & 6.09 & 5.71 \\ \hline
%\end{tabular}
%\end{center}
%\end{table}

\begin{table}[htbp]
\begin{center}
\caption{Table of values of the Lipchitz bound $\Delta_{\bm{\alpha},\mathbf{x}}$, of the synthesizer under the $\bm{\alpha}-$weighted pseudo posterior mechanism, for a series of $(c, g)$ configurations. $\Delta_{Unweighted} = 78.7.$}
\label{tab:cg}
\begin{tabular}{ r  c }
\hline
$(c, g)$  & $\Delta_{\bm{\alpha},\mathbf{x}}$ value \\ \hline
(0.7, 0.0) & 10.10 \\
(0.6, 0.0) & 8.16 \\
(0.6, -0.1) & 7.30 \\
(0.5, 0.0) & 6.09 \\
(0.5, -0.1) & 5.71 \\
(0.5, -0.3) & 2.25 \\ \hline
\end{tabular}
\end{center}
\end{table}

%
%\begin{figure}[H]
%\centering
%\includegraphics[width=0.9\textwidth]{mean_compare_scaleshift_lg}
%\caption{mean}
%\label{fig:meancompare}
%\end{figure}
%
%
%\begin{figure}[H]
%\centering
%\includegraphics[width=0.9\textwidth]{q90_compare_scaleshift_lg}
%\caption{90 quantile}
%\label{fig:q90compare}
%\end{figure}

\section{Conclusion}
\label{sec:conclusion}
This paper adapts the $\bm{\alpha}-$weighted pseudo posterior synthesizer and converts any non-private posterior synthesizer to a formally private mechanism. Our pseudo posterior mechanism provides a much higher utility than the EM for equivalent local privacy guarantee, $\epsilon_{\mathbf{x}}$, due to its surgical downweighting of high-risk records (as opposed to the scalar downweighting imposed by the EM).  The construction for the $\bm{\alpha}-$weighted pseudo posterior mechanism utilizes the log-pseudo likelihood to develop the local Lipschitz bound. We provide an asymptotic result on the contraction of a local Lipschitz to a global bound that guarantees an $(\epsilon,\pi)-$probabilistic DP guarantee where $\pi$ contracts onto $0$ for $n$ sufficiently large.  We are able to increase the rate of contraction by truncating the weight, $\alpha^{\ast}_{i} = 0 $ if the weighted log likelihood contribution, $\alpha_{i} \times f_{\theta}(x_{i}) > M$, where $M$ becomes the targeted global point of contraction.

Our $\bm{\alpha}-$weighted pseudo posterior mechanism has the feature that it accommodates any synthesizer model formulated by the statistical agency and offers a simple weighting scheme that guarantees a pDP result.  The simple weighting allows the posterior sampling scheme devised for the non-private synthesizer to be utilized for synthesis with minor modification for the $\bm{\alpha}-$weighted pseudo posterior mechanism.

%\section*{Acknowledgements}
%
%This research is supported by ASA/NSF/BLS Senior Research Fellow Program.

\bibliography{DPbib}

\newpage
\appendix %switches to alphabetical sections - comment out to go back to numbered sections

\section{Proofs for Theoretical Results in Sections~\ref{sec:theory} and \ref{sec:lockin}}\label{app:theory}
%using Assumption 1 show \delta or 2 delta
%using Assumption 2 show proptionality constant K and K or 2 K

\subsection{Proof for Theorem~\ref{th:dpresult}}

We begin by stating an enabling result that connects the global Lipschitz bound, $\Delta_{\bm{\alpha}}$, to the KL divergence between the posterior densities (given $\mathbf{x}$ versus $\mathbf{y}$) from the inclusion of a database record.
\begin{theorem}
\label{thm:KL}
$\forall \mathbf{x} \in \mathcal{X}^{n}, \mathbf{y} \in \mathcal{X}^{n-1}: \delta(\mathbf{x}, \mathbf{y}) = 1$ and $\bm{\alpha}(\cdot)$ with $\Delta_{\bm{\alpha}} > 0$ satisfying Assumption \ref{ass:lipschitz},
\begin{equation}
\mathop{\sup}_{\mathbf{x} \in \mathcal{X}^{n}, \mathbf{y}\in \mathcal{X}^{n-1}: \delta(\mathbf{x}, \mathbf{y}) = 1} D_{KL} \left[\infdiv{\xi^{\bm{\alpha}(\mathbf{x})} (\cdot \mid \mathbf{x})}{\xi^{\bm{\alpha}(\mathbf{y})}(\cdot \mid \mathbf{y})}\right] \leq 2\Delta_{\bm{\alpha}},
\end{equation}
where $D_{KL}(\infdiv{P}{Q}) = \int_{\mathcal{X}^{n}} \ln \frac{dP}{d Q}dP$.
\end{theorem}

\begin{proof}
\begin{eqnarray}
D_{KL} \left[\infdiv{\xi^{\bm{\alpha}(\mathbf{x})} (\cdot \mid \mathbf{x})}{\xi^{\bm{\alpha}(\mathbf{y})}(\cdot \mid \mathbf{y})}\right]
&=& \int_{\Theta} \ln \frac{d \xi^{\bm{\alpha}(\mathbf{x})}(\theta \mid \mathbf{x})}{d \xi^{\bm{\alpha}(\mathbf{y})}(\theta \mid \mathbf{y})} d\xi^{\bm{\alpha}(\mathbf{x})}(\theta \mid \mathbf{x}) \nonumber \\
&=& \int_{\Theta} \ln \frac{p_{\theta}^{\bm{\alpha}(\mathbf{x})}(\mathbf{x})}{p_{\theta}^{\bm{\alpha}(\mathbf{y})}(\mathbf{y})} d \xi^{\bm{\alpha}(\mathbf{x})}(\theta \mid \mathbf{x}) + \int_{\Theta} \ln \frac{\phi^{\bm{\alpha}(\mathbf{y})}(\mathbf{y})}{\phi^{\bm{\alpha}(\mathbf{x})}(\mathbf{x})} d\xi^{\bm{\alpha}(\mathbf{x})} (\theta \mid \mathbf{x}) \nonumber \\
&\leq& \int_{\Theta} \abs{\ln \frac{p_{\theta}^{\bm{\alpha}(\mathbf{x})}(\mathbf{x})}{p_{\theta}^{\bm{\alpha}(\mathbf{y})}(\mathbf{y})}} d \xi^{\bm{\alpha}(\mathbf{x})} (\theta \mid \mathbf{x}) + \int_{\Theta} \ln \frac{\phi^{\bm{\alpha}(\mathbf{y})}(\mathbf{y})}{\phi^{\bm{\alpha}(\mathbf{x})}(\mathbf{x})} d\xi^{\bm{\alpha}(\mathbf{x})} (\theta \mid \mathbf{x}) \nonumber \\
&\leq& \Delta_{\bm{\alpha}} + \abs{\ln \frac{\phi^{\bm{\alpha}(\mathbf{y})}(\mathbf{y})}{\phi^{\bm{\alpha}(\mathbf{x})}(\mathbf{x})}}
\end{eqnarray}

From Assumption \ref{ass:lipschitz}, $p_{\theta}^{\bm{\alpha}(\mathbf{x}) } (\mathbf{x}) \leq \exp (\Delta_{\bm{\alpha}})p_{\theta}^{\bm{\alpha}(\mathbf{y}) }(\mathbf{y}), \forall \theta \in \Theta$, so
\begin{equation}
\phi^{\bm{\alpha}(\mathbf{y})}(\mathbf{y}) = \int_{\Theta} p_{\theta}^{\bm{\alpha}(\mathbf{y})}(\mathbf{y}) d \xi(\theta) \leq \exp(\Delta_{\bm{\alpha}}) \int_{\Theta} p_{\theta}^{\bm{\alpha}(\mathbf{x})}(\mathbf{x}) d \xi (\theta) = \exp (\Delta_{\bm{\alpha}}) \phi^{\bm{\alpha}}(\mathbf{x}),
\end{equation}
which gives
\begin{equation}
\mathop{\sup}_{\mathbf{x} \in \mathcal{X}^{n}, \mathbf{y} \in \mathcal{X}^{n-1}: \delta(\mathbf{x}, \mathbf{y}) = 1} D_{KL} \left[\infdiv{\xi^{\bm{\alpha}(\mathbf{x})} (\cdot \mid \mathbf{x})}{\xi^{\bm{\alpha}(\mathbf{y})}(\cdot \mid \mathbf{y})}\right] \leq 2\Delta_{\bm{\alpha}}.
\end{equation}
\end{proof}

\subsubsection{Proof of Theorem~\ref{th:dpresult}}

From Assumption \ref{ass:lipschitz}, $\frac{p_{\theta}^{\bm{\alpha}(\mathbf{x})}(\mathbf{x})}{p_{\theta}^{\bm{\alpha}(\mathbf{y})}(\mathbf{y})} \leq \exp(\Delta_{\bm{\alpha}})$. From Theorem \ref{thm:KL}, we show $\phi^{\bm{\alpha}(\mathbf{y})}(\mathbf{y}) \leq \exp(\Delta_{\bm{\alpha}}) \phi^{\bm{\alpha}(\mathbf{x})}(\mathbf{x})$. Then, $\forall \mathbf{x} \in \mathcal{X}^{n}$ and for each $\mathbf{x}$, $\forall \mathbf{y}\in\mathcal{X}^{n-1}:\delta(\mathbf{x},\mathbf{y}=1)$,
\begin{eqnarray}
\xi^{\bm{\alpha}(\mathbf{x})}(B \mid \mathbf{x})
&=& \frac{\int_B \frac{p_{\theta}^{\bm{\alpha}(\mathbf{x})}(\mathbf{x})} {p_{\theta}^{\bm{\alpha}(\mathbf{y})}(\mathbf{y})} p_{\theta}^{\bm{\alpha}(\mathbf{y})}(\mathbf{y}) d\xi(\theta)}{\phi^{\bm{\alpha}(\mathbf{y})}(\mathbf{y})} \cdot \frac{\phi^{\bm{\alpha}(\mathbf{y})}(\mathbf{y})}{\phi^{\bm{\alpha}(\mathbf{x})}(\mathbf{x})} \nonumber\\
&\leq& \exp(2\Delta_{\bm{\alpha}}) \xi^{\bm{\alpha}(\mathbf{y})}(B \mid \mathbf{y}).
\end{eqnarray}

\subsection{Proof for Lemma~\ref{lm:postpred}}
\begin{eqnarray}
P^{\bm{\alpha}(\mathbf{x})}(\bm{\zeta} \in C \mid \mathbf{x})
&=& \int P(\bm{\zeta} \in C \mid \mathbf{x}, \theta) d\xi^{\bm{\alpha}(\mathbf{x})}(\theta \mid \mathbf{x}) \nonumber \\
&=& \int P(\bm{\zeta} \in C \mid \theta) d\xi^{\bm{\alpha}(\mathbf{x})}(\theta \mid \mathbf{x}) \nonumber \\
&=& \int P(\bm{\zeta} \in C \mid \theta) \frac{d\xi^{\bm{\alpha}(\mathbf{x})}(\theta \mid \mathbf{x})}{d\xi^{\bm{\alpha}(\mathbf{y})}(\theta \mid \mathbf{y})} d\xi^{\bm{\alpha}(\mathbf{y})}(\theta \mid \mathbf{y}) \nonumber \\
&\leq& e^{\epsilon} \int P(\bm{\zeta} \in C \mid \theta) d\xi^{\bm{\alpha}(\mathbf{y})}(\theta \mid \mathbf{y}) \nonumber \\
&=& e^{\epsilon} P^{\bm{\alpha}(\mathbf{y})} (\bm{\zeta} \in C \mid \mathbf{y}).
\end{eqnarray}

\subsection{Proof of Theorem~\ref{th:consistency}}
Let us the define the following subset of $\theta\in\Theta$,
\begin{equation*}
  U_{n} = \left\{\theta\in\Theta: \left[(1-\alpha_{m})D^{(n_{A})}_{\theta_{0},\bm{\alpha}}\left(\theta,\theta^{\ast}\right) + (1-\alpha^{(n)})D^{(n_{Q})}_{\theta_{0},1^{-}}\left(\theta,\theta^{\ast}\right)\right] \geq (D+3t) n\tau_{n}^{2}\right\},
\end{equation*}
which is the restricted set for which we will bound the pseudo posterior distribution, $\xi^{\bm{\alpha}}\left(U_{n}\mid \mathbf{x}\right)$, from above to achieve the result of Theorem~\ref{th:consistency}.   We begin with the statement and proof
of Lemma~\ref{denominator} that extends Lemma 8.1 of  \citet{Ghosal00convergencerates} to our $\alpha-$pseudo posterior in order to provide a concentration
inequality to probabilistically (in $P_{\theta_{0}}-$probability) bound the denominator of the $\bm{\alpha}-$pseudo posterior distribution, $\xi^{\bm{\alpha}}\left(U_{n}\mid \mathbf{x}\right)$, from below.

\subsubsection{Enabling Lemma}
%move on to second lemma
\begin{lemma}\label{denominator}
(Concentration Inequality)
Suppose Assumption~\ref{prior} holds. Define $\alpha_{m} = \mathop{\max}_{i \in A_{n}}\alpha_{i}$ and $\alpha_{l} = \mathop{\min}_{i \in A_{n}}\alpha_{i}$.
For every $\tau_{n} > 0$ and measure $\Pi$ on the set $B_{n}\left(\theta^{\ast},\xi;\theta_{0}\right)$, we have for every $C^{\ast}_{1} = \sqrt{2+C_{1}^{2}+C_{3}^{2}}$, and $n$ sufficiently large,
\begin{equation}\label{denomresult}
P_{\theta_{0}}\left\{\mathop{\int}_{\theta\in B_{n}}\displaystyle e^{-r_{n,\bm{\alpha}}\left(\theta,\theta^{\ast}\right)}
\xi\left(d\theta\right)\leq e^{-\alpha_{m}(D+t)n\tau_{n}^{2}}\right\}
\leq \frac{(1+\alpha_{l}^{2})(C_{1}^{\ast})^{2}}{\alpha_{m}^{2}}\times\frac{1}{(D+t-1)^{2}n\tau_{n}^{2}},
\end{equation}
where the above probability is taken with the respect to $P_{\theta_{0}}$.
\end{lemma}

\begin{proof}\label{AppDenominator}
The proof follows that of \citet{2015arXiv150707050S} by bounding the probability expression on left-hand size of Equation~(\ref{denomresult}).  We construct an $\bm{\alpha}-$weighted empirical distribution that we will need for the proof with,
\begin{equation}
\mathbb{P}_{n,\bm{\alpha}} = \frac{1}{n}\mathop{\sum}_{i=1}^{n}\alpha_{i}\delta\left(x_{i}\right),
\end{equation}
where $\delta(x_{i})$ denotes the Dirac delta function with probability mass $1$ at $x_{i}$.  We construct the associated scaled and centered empirical process, $\mathbb{G}_{n,\bm{\alpha}} = \sqrt{n}\left(\mathbb{P}_{n,\bm{\alpha}}-P_{\theta_{0}}\right)$.
The usual equally-weighted empirical distribution, $\mathbb{P}_{n}=\frac{1}{n}\mathop{\sum}_{i=1}^{n}\delta\left(x_{i}\right)$ and associated, $\mathbb{G}_{n} = \sqrt{n}\left(\mathbb{P}_{n}-P_{\theta_{0}}\right)$ may be viewed as special cases.   We may define the associated expectation functionals with respect to the $\bm{\alpha}-$weighted empirical distribution by $\mathbb{P}_{n,\bm{\alpha}} g = \frac{1}{n}\mathop{\sum}_{i=1}^{n}\alpha_{i}g\left(x_{i}\right)$.

Using Jensen's inequality,
\begin{equation}
\begin{split}\label{leftbound}
&\log\mathop{\int}_{\theta\in B_{n}}\mathop{\prod}_{i=1}^{n}\left[\frac{p_{\theta_{i}}}{p_{\theta^{\ast}_{i}}}\left(X_{i}\right)\right]^{\alpha_{i}}\xi\left(d\theta\right)\\
&\geq \mathop{\sum}_{i=1}^{n}\mathop{\int}_{\theta\in B_{n}}\alpha_{i}\log\frac{p_{\theta_{i}}}{p_{\theta^{\ast}_{i}}}\xi\left(d\theta\right) \\
&= n\mathbb{P}_{n,\bm{\alpha}}\mathop{\int}_{\theta\in B_{n}}\log\frac{p_{\theta}}{p_{\theta^{\ast}}}\xi\left(d\theta\right)
\end{split}
\end{equation}

We may use the above to now bound the left-hand size of Equation~(\ref{denomresult})
\begin{subequations}\label{leftboundcont}
\begin{align}
&P_{\theta_{0}}\left\{\mathop{\int}_{\theta\in B_{n}}\displaystyle e^{-r_{n,\bm{\alpha}}\left(\theta,\theta^{\ast}\right)}
\xi\left(d\theta\right)\leq e^{-\alpha_{m}(D+t)n\tau_{n}^{2}}\right\}\\
&\leq P_{\theta_{0}}\left\{n\mathbb{P}_{n,\bm{\alpha}}\mathop{\int}_{\theta\in B_{n}}\log\frac{p_{\theta}}{p_{\theta^{\ast}}}\xi\left(d\theta\right)\leq -\alpha_{m}(D+t)n\tau_{n}^{2}\right\}\\
&= P_{\theta_{0}}\left\{\mathbb{G}_{n,\bm{\alpha}}\mathop{\int}_{\theta\in B_{n}}\log\frac{p_{\theta}}{p_{\theta^{\ast}}}\xi\left(d\theta\right)\leq -\alpha_{m}(D+t)n\tau_{n}^{2}-\sqrt{n}P_{\theta_{0}}\log\frac{p_{\theta}}{p_{\theta^{\ast}}}\xi\left(d\theta\right)\right\}\\
&\leq P_{\theta_{0}}\left\{\mathbb{G}_{n,\bm{\alpha}}\mathop{\int}_{\theta\in B_{n}}\log\frac{p_{\theta}}{p_{\theta^{\ast}}}\xi\left(d\theta\right)\leq -\alpha_{m}(D+t)\sqrt{n}\tau_{n}^{2}-\sqrt{n}\tau_{n}^{2}\right\}\label{usepriorcond}\\
&= P_{\theta_{0}}\left\{\mathbb{G}_{n,\bm{\alpha}}\mathop{\int}_{\theta\in B_{n}}\log\frac{p_{\theta}}{p_{\theta^{\ast}}}\xi\left(d\theta\right)\leq -\alpha_{m}(D+t-1)\sqrt{n}\tau_{n}^{2}\right\},
\end{align}
\end{subequations}
where the bound in Equation~(\ref{usepriorcond}) uses the prior mass result from Assumption~\ref{prior}.  We proceed to use Chebyshev to bound the resultant probability, as follows:

\begin{align}
&P_{\theta_{0}}\left\{\mathbb{G}_{n,\bm{\alpha}}\mathop{\int}_{\theta\in B_{n}}\log\frac{p_{\theta}}{p_{\theta^{\ast}}}\xi\left(d\theta\right)\leq -\alpha_{m}(D+t-1)\sqrt{n}\tau_{n}^{2}\right\}\nonumber\\
&\leq\frac{\mathop{\int}_{\theta\in B_{n}}\left[\mathbb{E}_{P_{\theta_{0}}}
\left(\mathbb{G}_{n,\bm{\alpha}}\log\frac{p_{\theta}}{p_{\theta^{\ast}}}\right)^{2}\right]\xi\left(d\theta\right)}
{\alpha_{m}^{2}(D+t-1)^{2}n\tau_{n}^{4}}\label{chebyshev:e2},
\end{align}
where we have applied Fubini to the right side of Equation~(\ref{chebyshev:e2}) to move the expectation through the integral.  We now proceed to further bound the expression in brackets on the right-hand side of Equation~(\ref{chebyshev:e2}) from above. We
may decompose the expectation, as follows
\begin{equation}\label{gbound}
\mathbb{E}_{P_{\theta_{0}}}
\left(\mathbb{G}_{n,\bm{\alpha}}\log\frac{p_{\theta}}{p_{\theta^{\ast}}}\right)^{2}
\leq n\mathbb{E}_{P_{\theta_{0}}}\left(\mathbb{P}_{n,\bm{\alpha}}\log\frac{p_{\theta}}{p_{\theta^{\ast}}}-\mathbb{P}_{n}\log\frac{p_{\theta}}{p_{\theta^{\ast}}}\right)^{2} + \mathbb{E}_{P_{\theta_{0}}}\left(\mathbb{G}_{n,\bm{\alpha}}\log\frac{p_{\theta}}{p_{\theta^{\ast}}}\right)^{2}
\end{equation}
We first bound the second term on the right,
\begin{subequations}\label{runt}
\begin{align}
&\mathbb{E}_{P_{\theta_{0}}}\left(\mathbb{G}_{n,\bm{\alpha}}\log\frac{p_{\theta}}{p_{\theta^{\ast}}}\right)^{2}\\
&\leq \mathbb{E}_{P_{\theta_{0}}}\left(\sqrt{n}\mathbb{P}_{n,\bm{\alpha}}\log\frac{p_{\theta}}{p_{\theta^{\ast}}}\right)^{2}\\
&\leq \mathbb{E}_{P_{\theta_{0}}}\left(\frac{1}{\sqrt{n}}\mathop{\sum}_{i=1}^{n}\log\frac{p_{\theta}}{p_{\theta^{\ast}}}\right)^{2}\\
&\leq \frac{1}{n}\mathop{\sum}_{i=1}^{n}\mathbb{E}_{P_{\theta_{0}}}\left(\log\frac{p_{\theta}}{p_{\theta^{\ast}}}\right)^{2}\\
&\leq \frac{1}{n}\times n\tau_{n}^{2} = \tau_{n}^{2},
\end{align}
\end{subequations}
where we use independence of the $X_{i}$ to establish the fourth equation and Assumption~\ref{prior} to achieve the fifth equation.

We proceed to further simplify the bound in the first term on the right in Equation~(\ref{gbound}):
\begin{subequations}\label{bound}
\begin{align}
  & n\mathbb{E}_{P_{\theta_{0}}}\left(\mathbb{P}_{n,\bm{\alpha}}\log\frac{p_{\theta}}{p_{\theta^{\ast}}}-\mathbb{P}_{n}\log\frac{p_{\theta}}{p_{\theta^{\ast}}}\right)^{2} \\
  & = n\mathbb{E}_{P_{\theta_{0}}}\left( \frac{1}{n}\mathop{\sum}_{i=1}^{n}\left(\alpha_{i}-1\right)\log\frac{p_{\theta_{i}}}{p_{\theta^{\ast}_{i}}}\right)^{2} \\
  & = \frac{1}{n}\mathop{\sum}_{i,j=1}^{n}\mathbb{E}_{P_{\theta_{0}}}\left[\left(\alpha_{i}-1\right)\left(\alpha_{j}-1\right)\log\frac{p_{\theta_{i}}}{p_{\theta^{\ast}_{i}}}\left(X_{i}\right)\log\frac{p_{\theta,j}}{p_{\theta^{\ast},j}}\left(X_{j}\right)\right]\\
  \begin{split}
  & = \frac{1}{n}\mathop{\sum}_{i=j=1}^{n}\mathbb{E}_{P_{\theta_{0}}}\left[\left(\alpha_{i}-1\right)^{2}\log\frac{p_{\theta_{i}}}{p_{\theta^{\ast}_{i}}}\left(X_{i}\right)^{2}\right]\\
  & + \frac{1}{n}\mathop{\sum}_{i\neq j=1}^{n}\mathbb{E}_{P_{\theta_{0}}}\left\lvert\left[\left(\alpha_{i}-1\right)\left(\alpha_{j}-1\right)\log\frac{p_{\theta_{i}}}{p_{\theta^{\ast}_{i}}}\left(X_{i}\right)\log\frac{p_{\theta,j}}{p_{\theta^{\ast},j}}\left(X_{j}\right)\right]\right\rvert
  \end{split}\\
  \begin{split}
  & \leq \frac{1}{n}\left\{\left(1-\alpha_{l}\right)^{2}\mathop{\sum}_{i\neq j=1}^{n}\mathbb{E}_{P_{\theta_{0}}}\left[\log\frac{p_{\theta_{i}}}{p_{\theta^{\ast}_{i}}}\left(X_{i}\right)^{2}\right] \right\}\\
  & + \frac{1}{n}\left(1-\alpha_{l}\right)^{2}\mathop{\sum}_{i\neq j\in A_{n}}\left\lvert\mathbb{E}_{P_{\theta_{0}}}\log\frac{p_{\theta_{i}}}{p_{\theta^{\ast}_{i}}}\left(X_{i}\right)\log\frac{p_{\theta,j}}{p_{\theta^{\ast},j}}\left(X_{j}\right)\right\rvert\\
  & + \frac{1}{n}\left(1-\alpha^{(n)}\right)^{2}\mathop{\sum}_{i\neq j\in Q_{n}}\left\lvert\mathbb{E}_{P_{\theta_{0}}}\log\frac{p_{\theta_{i}}}{p_{\theta^{\ast}_{i}}}\left(X_{i}\right)\log\frac{p_{\theta,j}}{p_{\theta^{\ast},j}}\left(X_{j}\right)\right\rvert\\
  \end{split}\\
  & \leq \frac{1}{n}\left\{\left(1-\alpha_{l}\right)^{2}n\tau_{n}^{2}\right\} + \frac{1}{n}\left(1-\alpha_{l}\right)^{2}\left(C_{1}^{2}n-C_{1}\sqrt{n}\right)\tau_{n}^{2} + n_{Q}\frac{C_{3}^{2}\tau_{n}^{2}}{n_{Q}}\label{prioronBn}\\
  & = \left\{\left(1-\alpha_{l}\right)^{2}\tau_{n}^{2}\right\} + \left(1-\alpha_{l}\right)^{2}C_{1}^{2}\tau_{n}^{2} + C_{3}^{2}\tau_{n}^{2},
\end{align}
\end{subequations}
for sufficiently large $n$.  The bound in Equation~(\ref{prioronBn}) results from the restriction of $\theta$ to $B_{n}\left(\theta^{\ast},\eta;\theta_{0}\right)$ and also from Assumption~\ref{size} that regulates the growth of the number of $\alpha_{i} < 1^{-}$ and the magnitude of $(1-\alpha^{(n)})$.

We may now bound the expectation on the right-hand size of Equation~(\ref{chebyshev:e2}),
\begin{subequations}
  \begin{align}
  \mathbb{E}_{P_{\theta_{0}}}
  \left(\mathbb{G}_{n,\bm{\alpha}}\log\frac{p_{\theta}}{p_{\theta^{\ast}}}\right)^{2} &\leq \left\{\left(1-\alpha_{l}\right)^{2}\tau_{n}^{2}\right\}\left(1-\alpha_{l}\right)^{2}C_{1}^{2}\tau_{n}^{2}+ \tau_{n}^{2} \\
  &\leq \left\{\left(1-2\alpha_{l} + \alpha_{l}^{2}\right)\tau_{n}^{2} + \left(1-2\alpha_{l}+\alpha_{l}^{2}\right)C_{1}^{2}\tau_{n}^{2} + C_{3}^{2}\eta{n}^{2} + \tau_{n}^{2}\right\}\\
  &\leq (2+C_{1}^{2}+C_{3}^{2})\tau_{n}^{2} + (1+C_{1}^{2})\alpha_{l}^{2}\tau_{n}^{2} \leq (1+\alpha_{l})^{2}(C_{1}^{\ast})^{2}\tau_{n}^{2} %\leq (1+\alpha_{m}^{2})(C_{1}^{\ast})^{2}\tau_{n}^{2},
  \end{align}
\end{subequations}
  for $n$ sufficiently large, where we set $C^{\ast}_{1} := \sqrt{C_{1}^{2} + C_{3}^{2} +2}$.  This concludes the proof.
\end{proof}

\subsubsection{Proof of Theorem~\ref{th:consistency}}
We begin by constructing the $\bm{\alpha}-$pseudo posterior distribution on the set, $U_{n}$,
\begin{equation}\label{postonU}
\xi^{\bm{\alpha}}\left(U_{n} \mid  \mathbf{x}\right) = \frac{\mathop{\int}_{U_{n}}e^{-r_{n,\bm{\alpha}}\left(\theta,\theta^{\ast}\right)}\xi(d\theta)}{\mathop{\int}_{\Theta}e^{-r_{n,\bm{\alpha}}\left(\theta,\theta^{\ast}\right)}\xi(d\theta)}.
\end{equation}
We next bound the numerator from above in $P_{\theta_{0}}-$ probability.

\begin{subequations}\label{numerator}
  \begin{align}
  &\mathbb{E}_{P_{\theta_{0}}}\mathop{\int}_{U_{n}}e^{-r_{n,\bm{\alpha}}\left(\theta,\theta^{\ast}\right)}\xi(d\theta)\\
  & = \mathop{\int}_{U_{n}}A^{(n)}_{\theta_{0},\bm{\alpha}}\left(\theta,\theta^{\ast}\right)\xi(d\theta)\label{Aform}\\
  & = \mathop{\int}_{U_{n}} e^{-\mathop{\sum}_{i=1}^{n}\left(1-\alpha_{i}\right)D_{\theta_{0},\bm{\alpha},i}}\xi(d\theta)\\
  & \leq \mathop{\int}_{U_{n}} e^{-\left(1-\alpha_{m}\right)\mathop{\sum}_{i\in A_{n}}D_{\theta_{0},\bm{\alpha},i} -\left(1-\alpha^{(n)}\right)\mathop{\sum}_{i\in Q_{n}}D_{\theta_{0},1^{-},i}}~\xi(d\theta)\label{alphabound}\\
  & \leq e^{-\left(D+3t\right)n\tau_{n}^{2}}\label{numbound},
  \end{align}
\end{subequations}
where we use Fubini to switch the order of expectation and integration in Equation~(\ref{Aform}).  We achieve the bound in Equation~(\ref{alphabound}) since
$D_{\theta_{0},\bm{\alpha},i} >0,~\forall i \in (1,\ldots,n)$ and \citet{BFP2019AS} shows that $D^{(n)}_{\theta_{0},1^{-}}\left(\theta,\theta^{\ast}\right)$ is finite and contracts on the KL divergence.  The final bound uses the definition of $U_{n}$.

We proceed to use the Markov inequality and the definition for $U_{n}$ to achieve the numerator bound with respect to $P_{\theta_{0}}-$probability,
\begin{subequations}\label{finalnumbound}
  \begin{align}
  &P_{\theta_{0}}\left\{\mathop{\int}_{U_{n}}e^{-r_{n,\bm{\alpha}}\left(\theta,\theta^{\ast}\right)}\xi(d\theta) \geq e^{-\left(D+2t\right)n\tau_{n}^{2}}\right\}\\
  &\leq\frac{e^{-\left(D+3t\right)n\tau_{n}^{2}}}{e^{-\left(D+2t\right)n\tau_{n}^{2}}}= e^{-tn\tau_{n}^{2}} \leq \frac{(1+\alpha_{l}^{2})(C_{1}^{\ast})^{2}}{\alpha_{m}^{2}(D-1+t)^{2}n\tau_{n}^{2}}.
  \end{align}
\end{subequations}

We, next, turn to bounding the denominator of Equation~(\ref{postonU}), from below.   Since,
\begin{equation*}
  \mathop{\int}_{\theta\in\Theta}e^{-r_{n,\bm{\alpha}}\left(\theta,\theta^{\ast}\right)}\xi(d\theta) \geq \mathop{\int}_{\theta \in B_{n}}e^{-r_{n,\bm{\alpha}}\left(\theta,\theta^{\ast}\right)}\xi(d\theta),
\end{equation*}
we may use the result of Lemma~\ref{denominator} in,
\begin{equation}\label{denombound}
P_{\theta_{0}}\left\{\mathop{\int}_{\theta\in\Theta}e^{-r_{n,\bm{\alpha}}\left(\theta,\theta^{\ast}\right)}\xi(d\theta)\geq e^{-\alpha_{m}(D+t)n\tau_{n}^{2}}\right\} > 1 -\frac{(1+\alpha_{l}^{2})(C_{1}^{\ast})^{2}}{\alpha_{m}^{2}(D-1+t)^{2}n\tau_{n}^{2}}.
\end{equation}

Finally, combining the results of Equations~(\ref{postonU}),~(\ref{finalnumbound}) and ~(\ref{denombound}): With probability at least $1- \left[2/(D+t-1)^{2}n\tau_{n}^{2}\times (1+\alpha_{l}^{2}(C_{1}^{\ast})^{2})/\alpha_{m}^{2}\right]$,
\begin{align*}
\begin{split}
\xi^{\bm{\alpha}}\left(\left[(1-\alpha_{m})D^{(n_{A})}_{\theta_{0},\bm{\alpha}}\left(\theta,\theta^{\ast}\right) + (1-\alpha^{(n)})D^{(n_{Q})}_{\theta_{0},1^{-}}\left(\theta,\theta^{\ast}\right)\right] \geq (D+3t)n\tau_{n}^{2}\big\vert \mathbf{x}\right) &\leq \\
& e^{-\left(D+2t\right)n\tau_{n}^{2}} e^{\alpha_{m}(D+t)n\tau_{n}^{2}}
\end{split}\\
& \leq e^{-tn\tau_{n}^{2}}
\end{align*}

\section{Unweighted, Non-private Synthesizer}
\label{appendix:synthesizer}

Our description of the unweighted, non-private synthesizer follows closely of that in \citet{HuSavitskyWilliams2021rebds}. To simulate partially synthetic data for the CE sample, where only the sensitive, continuous family income variable is synthesized, we propose using a flexible, parametric finite mixture synthesizer. %As shown in \citet{HuSavitskyWilliams2021rebds}, their truncated Dirichlet process (TDP) mixture synthesizer produces synthetic CE data with high utility, and can be used for partial synthesis of continuous variable(s) utilizing a number of available predictors.

Equation (\ref{eq:y}) and Equation (\ref{eq:z}) present the first two levels of the hierarchical parametric finite mixture synthesizer: $y_i$ is the logarithm of the family income for CU $i$, and $\bm{x}_i$ is the $R \times 1$ predictor vector for CU $i$. The finite mixture utilizes a hyperparameter for the maximum number of mixture components (i.e., clusters), $K$, that is to set to be over-determined to permit the flexible clustering of CUs.  A subset of CUs that are assigned to cluster, $k$, employ the same generating parameters for $y$, $(\beta^{\ast}_{k},\sigma^{\ast}_{k})$, that we term a ``location".  Locations, $(\bm{\beta}^*, \bm{\sigma}^*)$, and the $n \times 1$ vector of cluster indicators, $z_i \in (1, \cdots, K)$, are all sampled for each CU, $i \in (1,\ldots,n)$.
\begin{eqnarray}
	\label{eq:y} y_i \mid \mathbf{X}_i, z_i, \mathbf{B}^{\ast}, \bm{\sigma}^{\ast} &\sim& \textrm{Normal}(y_i \mid \mathbf{x}_i^{'}\bm{\beta}^{\ast}_{z_i} , \sigma^{\ast}_{z_i}), \\
	\label{eq:z} z_i \mid \mathbf{\pi} &\sim& \textrm{Multinomial}(1; \pi_1, \cdots, \pi_{K}),
\end{eqnarray}
where the $K \times R$ matrix of regression locations, $\mathbf{B}^{\ast} = \left(\bm{\beta}^{\ast}_{1},\ldots,\bm{\beta}^{\ast}_{K}\right)^{'}$, denote cluster-indexed regression coefficients for $R$ predictors.  The $(\pi_{1},\ldots,\pi_{K})$ are, in turn, assigned a sparsity inducing Dirichlet distribution with hyperparameters specified as $\alpha/K$ for $\alpha \in \mathbb{R}^{+}$. We next describe our prior specification.

We induce sparsity in the number of clusters with,
\begin{align}
\left(\pi_{1},\ldots,\pi_{K}\right) &\sim \mbox{Dirichlet}\left(\frac{\alpha}{K},\ldots,\frac{\alpha}{K}\right)\label{eq:DPprior-pi}, \\
\alpha &\sim \textrm{Gamma}(a_{\alpha}, b_{\alpha}).	\label{eq:DPprior-alpha}
\end{align}

We specify multivariate Normal priors for each regression coefficient vector of coefficient locations, $\bm{\beta}^{\ast}_k$,
\begin{equation}
\bm{\beta}^{\ast}_k \iid \textrm{MVN}_{R}(\mathbf{0}, \mbox{diag}(\bm{\sigma}_{\beta})\times \mathop{\Omega_{\beta}}^{R \times R} \times \mbox{diag}(\bm{\sigma}_{\beta}) ), \label{eq:prior-beta}
\end{equation}
where the $R \times R$ correlation matrix, $\Omega_{\beta}$, receives a uniform prior over the space of $R \times R$ correlation matrices, and each component of $\bm{\sigma}_{\beta}$ receives a Student-t prior with $3$ degrees of freedom,
\begin{equation}
 \sigma^{\ast}_{k} \iid \textrm{t}(3, 0, 1). \label{eq:prior-sigma}
\end{equation}

We proceed to describe how to generate partially synthetic data for the CE sample. To implement the finite mixture synthesizer, we first generate sample values of $(\bm{\pi}^{(l)}, \bm{\beta}^{*,(l)}, \bm{\sigma}^{*, (l)})$ from the posterior distribution at MCMC iteration $l$. Second, for CU $i$, we generate cluster assignments, $z_i^{(l)}$, from its full conditional posterior distribution given in \citet{HuSavitskyWilliams2021rebds} using the posterior samples of $\bm{\pi}^{(l)}$.  Lastly, we generate synthetic family income for CU $i$, $y_i^{\ast, (l)}$, from Equation (\ref{eq:y}) given $\bm{x}_i$, and samples of $z_i^{(l)}, \bm{\beta}^{*,(l)}$ and $\bm{\sigma}^{*, (l)}$. We perform these draws for all $n$ CUs, and obtain a partially synthetic database, $\bm{Z}^{(l)}$ at MCMC iteration $l$. We repeat this process for $m$ times, creating $m$ independent partially synthetic databases $\bm{Z} = (\bm{Z}^{(1)}, \cdots, \bm{Z}^{(m)})$.

\end{document}